\documentclass[acmsmall,natbib=false,screen]{acmart}
\pdfoutput=1

\setcopyright{rightsretained}
\acmDOI{10.1145/3704861}
\acmYear{2025}
\acmJournal{PACMPL}
\acmVolume{9}
\acmNumber{POPL}
\acmArticle{25}
\acmMonth{1}
\received{2024-07-09}
\received[accepted]{2024-11-07}

\RequirePackage[datamodel=acmdatamodel,style=acmauthoryear,backend=biber]{biblatex}
\DeclareFieldFormat{doi}{%
  \mkbibacro{DOI}\addcolon\space
  \ifhyperref%
    {\href{https://doi.org/#1}{\nolinkurl{https://doi.org/#1}}}
    {\nolinkurl{https://doi.org/#1}}}

\usepackage{mathpartir}

\newcommand{\At}{\mathsf{At}}

\newcommand{\myrightarrow}[1]{\mathrel{\raisebox{-3pt}{$\xrightarrow{#1}$}}}
\newcommand{\transition}[2]{\myrightarrow{#1|#2}}

\newcommand{\GKAT}{\mathsf{GKAT}}
\newcommand{\KAT}{\mathsf{KAT}}
\newcommand{\BA}{\mathsf{BA}}
\newcommand{\sem}[1]{\llbracket#1\rrbracket}
\newcommand{\G}{\mathcal{G}}

\usepackage{xspace}

\newcommand{\apply}{\textup{apply}\xspace}
\newcommand{\compose}{\textsf{compose}\xspace}

\newcommand{\ssigma}{\mathfrak{s}}
\newcommand{\ttau}{\mathfrak{t}}

\newcommand{\Ifthenelse}[3]{\textnormal{\textsf{if $#1$ then $#2$ else $#3$}}}
\newcommand{\while}[2]{\textnormal{\textsf{while $#1$ do $#2$}}}

\AtEndPreamble{%
    \theoremstyle{acmplain}
    
    \theoremstyle{acmdefinition}
    \newtheorem{remark}[theorem]{Remark}
    
}

\usepackage{tikz}
\usetikzlibrary{arrows}
\usetikzlibrary{automata}
\usetikzlibrary{positioning}

\usepackage[nameinlink]{cleveref}

\begin{document}

\title{Algebras for Deterministic Computation Are Inherently Incomplete}

\author{Balder ten Cate}
\orcid{0000-0002-2538-5846} 
\email{b.d.tencate@uva.nl}
\affiliation{%
  \institution{ILLC, University of Amsterdam}
  \city{Amsterdam}
  \country{The Netherlands}
}

\author{Tobias Kappé}
\orcid{0000-0002-6068-880X} 
\email{t.w.j.kappe@liacs.leidenuniv.nl}
\affiliation{%
  \institution{LIACS, Leiden University}
  \city{Leiden}
  \country{The Netherlands}
}
\authornote{Work performed at Open University of the Netherlands and ILLC, University of Amsterdam}

\begin{abstract}
Kleene Algebra with Tests (KAT) provides an elegant algebraic framework for describing non-deterministic finite-state computations.
Using a small finite set of non-deterministic programming constructs (sequencing, non-deterministic choice, and iteration) it is able to express all \emph{non-deterministic} finite state control flow over a finite set of primitives.
It is natural to ask whether there exists a similar finite set of constructs that can capture all \emph{deterministic} computation.
We show that this is not the case.
More precisely, the deterministic fragment of KAT is not generated by \emph{any} finite set of regular control flow operations.
This generalizes earlier results about the expressivity of the traditional control flow operations, i.e., sequential composition, if-then-else and while.
\end{abstract}



\maketitle

\section{Introduction}
A classic theorem by \citet{BohmJacopini1966} states that the control flow operations of \emph{sequential composition}, \emph{if-then-else} and \emph{while}, in combination with variable assignments of the form $x := c$ and variable equality checks of the form $x = c$ (for $c$ a constant), suffice to specify all types of control flow, including those built using non-local control flow constructs such as \emph{goto} and \emph{break}.

\citet{KozenTseng} showed that the same three control flow operations are \emph{not} sufficient to capture all deterministic regular computations without the use of variables.
More formally, their result can be interpreted to mean that the expressivity of \emph{Guarded Kleene Algebra with Tests} (GKAT)~\cite{GKAT} is strictly contained in the deterministic fragment of \emph{Kleene Algebra with Tests} (KAT)~\cite{KAT}.
We will show that this remains true when GKAT is extended with finitely many additional regular control flow operations.
Our main contributions are as follows:
\begin{enumerate}
    \item
    We introduce the notion of a (deterministic) \emph{regular control flow operation}, which generalizes familiar program compositions to include deterministic control flow operations not expressible in GKAT, such as \emph{repeat-while-changes}, while excluding non-local control flow.
    \item
    We consider extensions of GKAT with finitely many regular control flow operations.
    Our main result is that, for every finite set $\mathcal{O}$ of such operations, there exists a regular control flow operation that cannot be expressed by compositions of operations in $\mathcal{O}$.
    \item
    We identify an infinite but nontrivial generating set of regular control flow operations.
    \item
    We also show that deciding whether a control flow operation is deterministic is coNP-complete.
    For bounded test alphabets, the problem is solvable in polynomial time.
\end{enumerate}

We believe that the results above offer insights into the nature of deterministic control flow, as well as a new style of analysis of interest to programming language researchers.
Specifically:
\begin{itemize}
\item Prior work on the expressive power of control flow operations~\cite{peterson-kasami-tokura-1973,Kosaraju} focuses on known primitives (\emph{if-then}, \emph{while-do}, \emph{break}, \emph{goto}).
By comparison, our results provide a broader perspective: not only is while-based control flow incomplete, so are all languages consisting of regular control flow operations.
This negatively resolves a question of \citet{bogaerts2024preservation}, namely whether the deterministic fragment of KAT is generated by finitely many term-definable operations.

\item Our results demonstrate that KAT and GKAT, with their relatively simple theory, provide a useful vehicle to formalize and answer questions about (deterministic) control flow.
We believe our work also provides a good starting point for posing and studying different kinds of (non-local) flow control, such as the \emph{continue} statement, exceptions, and algebraic effects.

\item
KAT and GKAT are useful tools for program analysis and optimization through decidable equational reasoning~\cite{kozen-patron-2000,anderson-foster-guha-etal-2014}.
However, our understanding of GKAT is impeded by the absence of a characterization of its automata~\cite{GKAT,GKATBisim}.
Our main result provides a fresh angle on this question.
\end{itemize}

\paragraph{Outline}
\Cref{sec:overview} provides context and motivation.
\Cref{sec:kat} reviews Kleene Algebra with Tests.
In \Cref{sec:determinism} we study \emph{determinism} properties for KAT expressions.
\Cref{sec:regular-operations} introduces the notion of a \emph{regular control flow operation}.
In \Cref{sec:main} we prove our main result, i.e., item (2) above.
In \Cref{sec:further} we discuss some variations of our setting, and related and future work in \Cref{sec:discussion}.

\section{Overview}%
\label{sec:overview}

In this section, we attempt to give an overview of our main thesis, and build intuition for its interpretation along the way.
We start by looking at (the skeleton of) a very basic programming language, and work towards a semi-formal statement of our results from there.

\subsection{An Abstract Programming Language}

\emph{Guarded Kleene Algebra with Tests}, or \emph{GKAT} for short~\cite{KozenTseng,GKAT}, is an abstraction of a \textsc{while}-like programming language, where primitive actions are given by a set $\Sigma$, and primitive tests by a set $T$.
Concretely, the set of \emph{GKAT programs} $\GKAT$ is generated by:
\begin{align*}
    \BA \ni b, c &::= \mathsf{false} \mid \mathsf{true} \mid t \in T \mid b\ \mathsf{or}\ c \mid b\ \mathsf{and}\ c \mid \mathsf{not}\ b \tag{tests}\\
    \GKAT \ni e, f &::= \mathsf{assert}\ b \mid p \in \Sigma \mid e; f \mid \Ifthenelse{b}{e}{f} \mid \while{b}{e} \tag{programs}
\end{align*}
If we want to give a denotational semantics to such programs, we need to say how the semantics of a composed program arises from its components.
Let us start at the level of tests.
If we model the possible states of our machine with a set $S$, then the natural semantics of a test would be the set of all states in which it is true.
Given a function $\tau: T \to 2^S$ that does this for each primitive test, we can compute such a semantics in a straightforward manner, as follows.
\begin{align*}
    \sem{\mathsf{false}}_\tau &= \emptyset
        & \sem{t}_\tau &= \tau(t)
        & \sem{b\ \mathsf{or}\ c}_\tau &= \sem{b}_\tau \cup \sem{c}_\tau \\
    \sem{\mathsf{true}}_\tau &= S
        & \sem{\mathsf{not}\ b}_\tau &= S \setminus \sem{b}_\tau
        & \sem{b\ \mathsf{and}\ c}_\tau &= \sem{b}_\tau \cap \sem{c}_\tau
\end{align*}

Defining the semantics of GKAT programs proper is somewhat less straightforward, but still doable.
A natural way to go about it is to model program semantics in terms of a function $f$ on machine states, which represents the effected computation: if the machine is in state $s$ when the program starts, it terminates in state $f(s)$ when the program halts.
This model is usually extended to partial functions, where $f(s)$ being undefined is interpreted to mean that the program did not terminate successfully --- either by getting stuck in a loop, or by running into an error.

Suppose that for each primitive action $p \in \Sigma$ we have a partial function $\sigma(p): S \rightharpoonup S$ that tells us whether (and how) $p$ affects the machine's state, or fails.
We can then give a semantics of programs in terms of partial functions parameterized by $\tau$ and $\sigma$.
The semantics of an assertion just checks whether the current test satisfies the state; the semantics of a primitive action refers to $\sigma$.
\begin{mathpar}
    \sem{\mathsf{assert}\ b}_\tau^\sigma(s)
        = s \quad\text{(if $s \in \sem{b}_\tau$)}
    \and
    \sem{p}_\tau^\sigma(s)
        = \sigma(p)(s)
\end{mathpar}
The semantics of the $\Ifthenelse{b}{e}{f}$ construct simply checks which subprogram needs to be executed based on the test.
Sequential composition is achieved by composition of partial functions.
\begin{mathpar}
    \sem{\Ifthenelse{b}{e}{f}}_\tau^\sigma(s)
        = \begin{cases}
          \sem{e}_\tau^\sigma(s) & s \in \sem{b}_\tau \\
          \sem{f}_\tau^\sigma(s) & s \not\in \sem{b}_\tau
          \end{cases}
    \and
    \sem{e; f}_\tau^\sigma(s)
        = \sem{f}_\tau^\sigma(\sem{e}_\tau^\sigma(s))
\end{mathpar}
Finally, the semantics of the $\while{b}{e}$ construct checks whether $b$ is satisfied by the current state; if not, then the output state is unchanged; if it is, then $e$ is applied to the input state, and the loop starts again.
Given this intuition, we can then define the semantics as the smallest partial function that satisfies the following recursive equation:
\begin{mathpar}
    \sem{\while{b}{e}}_\tau^\sigma(s)
        = \begin{cases}
          \sem{\while{b}{e}}_\tau^\sigma(\sem{e}_\tau^\sigma(s)) & s \in \sem{b}_\tau \\
          s & s \not\in \sem{b}_\tau
          \end{cases}
\end{mathpar}

\begin{remark}%
\label{remark:ill-founded}
Formally,  $\sem{\while{b}{e}}_\tau^\sigma$ is the least fixed point of the function that sends a partial function $f: S \rightharpoonup S$ to the partial function $f': S \rightharpoonup S$ such that $f'(s) = f(\sem{e}_\tau^\sigma(s))$ when $s \in \sem{b}_\tau$, and $f'(s) = s$ otherwise.
This is a Scott-continuous function on the directed-complete partial order on partial functions on $S$ where $f \leq g$ when for all $s \in S$ we have that $f(s) = g(s)$ whenever $f(s)$ is defined, and therefore the least fixed point exists uniquely by Kleene's fixpoint theorem.
\end{remark}

\subsection{The Road Less Traveled}
So far, GKAT and its denotational semantics have been fairly straightforward.
The only element that is somewhat non-standard is that we have not fixed a set of primitive tests and actions.
Instead, we have parametrized the semantics in terms of the interpretation of these primitives on some abstract set of machine states $S$.
Our language is thus purely about how the truth or falsehood of the tests determine which actions are executed, and in which order --- i.e., control flow.

At this point, we could choose to fix a set of machine states, introduce common programming language concepts like variables and I/O, and start studying the properties of the resulting language.
But instead, let us take advantage of this abstraction, and use GKAT as a vehicle for studying control flow.
If we adopt this point of view, then the semantics outlined above is best described as a dependently typed function, whose signature could be written as follows:
\[
    \sem{-}: \GKAT \to \forall S : \mathsf{Set}, (\underbrace{T \to 2^S}_{\text{test interp. $\tau$}}) \to  (\underbrace{\Sigma \to (S \rightharpoonup S)}_{\text{action interp. $\sigma$}}) \to (S \rightharpoonup S)
\]
In other words, $\sem{-}$ is a function that takes an abstract program $e$ from $\GKAT$, and returns for all sets $S$ a partial function $\sem{e}_\tau^\sigma$ on $S$, determined by an interpretation of the primitive tests $\tau: T \to 2^S$ and an interpretation of the primitive actions $\sigma: \Sigma \to (S \rightharpoonup S)$.
Intuitively, $\sem{e}$ is a description of the control flow implemented by $e$, and can determine the effect that control flow has as a function of the primitives.
The syntax of $\GKAT$ gives us a structured way to compose control flow.

\smallskip
One of the benefits of focusing on control flow is that it allows us to reason about equivalence of control flow expressions.
For instance, it should be clear that for all $e, f \in \GKAT$ and $b \in \BA$, we have that $\sem{\Ifthenelse{b}{e}{f}} = \sem{\Ifthenelse{\mathsf{not}\ b}{f}{e}}$,\footnote{%
    We assume dependent functional extensionality, so this asserts that $g = \Ifthenelse{b}{e}{f}$ and $h = \Ifthenelse{\mathsf{not}\ b}{f}{e}$ denote the same dependently typed function: for all sets $S$, $\tau: T \to 2^S$ and $\sigma: \Sigma \to (S \rightharpoonup S)$, it holds that $\sem{g}_\tau^\sigma = \sem{h}_\tau^\sigma$.
}
which is a property of conditional branches that programmers understand intuitively.
These equalities also go beyond what is easily provable based on the semantics; for instance, one can show that for all $e \in \GKAT$ and $b, c \in \BA$:
\[
    \sem{\while{b\ \mathsf{or}\ c}{e}}
        = \sem{\while{c}{e};\ \while{b}{(e;\ \while{c}{e})}}
\]
While proving these equalities by hand is possible, it can be cumbersome.
Fortunately, equivalence according to $\sem{-}$ is decidable --- as a matter of fact, it is even fixed-parameter tractable.
\begin{theorem}{\textnormal{\cite{GKAT}}}%
\label{theorem:gkat-checking}
Given $e, f \in \GKAT$, it is decidable whether $\sem{e} = \sem{f}$.
If we fix the number of primitive tests, it can be done in nearly linear time as a function of the size of $e$ and $f$.
\end{theorem}

The upshot, from a programmer's point of view, is that it is possible (and indeed feasible) to verify the correctness of certain program transformations, namely the ones that ``rearrange'' control flow constructs.
This does not contradict the fact that general program equivalence is undecidable, precisely because we abstract the primitives: GKAT has no data structures that could be used to simulate, e.g., a two-counter machine or a push-down automaton.

\subsection{Expressivity}
A natural question to ask, then, is what kind of control flow can be encoded in this language.
As a first and --- in our opinion reasonable --- guess, one might propose that this language should be able to express all control flow that is ``programmable''.
After all, we have the usual programming constructs at our disposal, so if we know how the primitive tests should affect the actions executed, we should be able to write a program in $\GKAT$ that implements this.
Surprisingly, this is not the case: there are types of control flow that could be built with specific primitives, but which cannot be the semantics of an expression in $\GKAT$.
The following demonstrates this more formally.

\begin{definition}%
\label{definition:repeat-while-changes-semantics}
Let $t \in T$ and $p \in \Sigma$.
We write $F$ for the program that iterates $p$ until the truth value of $t$ does not change.
Given an interpretation consisting of a set $S$, $\tau: T \to 2^S$, and $\sigma: \Sigma\to (S \rightharpoonup S)$, $F_\tau^\sigma$ can be specified using mutually recursive functions $u$ and $v$ on $S$ defined as follows:
\begin{mathpar}
    u(s) =
        \begin{cases}
        s & s \not\in \tau(t) \\
        v(\sigma(p)(s)) & s \in \tau(t)
        \end{cases}
    \and
    v(s) =
        \begin{cases}
        s & s \in \tau(t) \\
        u(\sigma(p)(s)) & s \not\in \tau(t)
        \end{cases}
\end{mathpar}
$F_\tau^\sigma$ is then given by $F_\tau^\sigma(s) = v(\sigma(p)(s))$ if $s \in \tau(t)$, and $F_\tau^\sigma(s) = u(\sigma(p)(s))$ otherwise. Formally, here, $v$ and $u$ are
 the least fixed point of the above system of equations (in the directed complete partial order of \emph{pairs} of partial functions on $S$), analogous to what was done in \Cref{remark:ill-founded}.
\end{definition}

It seems reasonable to expect at first that the program $F$ should be expressible in GKAT --- after all, it can be expressed as a partial function, which seems to be definable parameterized by an interpretation of tests and actions.
In that light, the following might be surprising.

\begin{theorem}%
\label{thm:impossible}
There is no GKAT expression $e$ such that $\sem{e} = F$.
\end{theorem}

Formally, \Cref{thm:impossible} follows from a result in~\cite[Appendix D]{GKATBisim}; we will explain more in Section~\ref{sec:determinism}.
A different program inexpressible in GKAT forms the main result of~\cite{KozenTseng}; part of the reason why this program is not expressible can already be found in~\cite{ashcroft.manna:translation,knuth-floyd-1971}, although this is not formalized in KAT\@.

\subsection{Possible Workarounds}

The intuition to~\Cref{thm:impossible} is that loops in GKAT do not allow us to express a loop halting condition like ``the truth value of $t$ did not change when we ran $p$''.
To achieve this, we would have to remember whether $t$ was true before executing $p$ again, for instance by setting a variable:
\[
    \begin{array}{l}
    \Ifthenelse{t}{x := 1}{x := 0}; \\
    p;\ \textsf{while}\ {(x = 0\ \mathsf{and}\ t)\ \mathsf{or}\ (x = 1\ \mathsf{and}\ \mathsf{not}\ t)}\ \textsf{do}
        \ \{\ {p;\ \Ifthenelse{t}{x := 1}{x := 0}}\ \}
    \end{array}
\]
The problem with this approach is that, in addition to assuming the existence of tests in $T$ like $x = 0$ and actions in $\Sigma$ like $x := 1$, this program can only be correct if we know that $t$ is independent of this variable, and that $p$ does not affect it.
By extension, equivalence of the above to any other program is really only meaningful if the other program is interpreted under the same assumptions.
This means that equivalence checking in the style of \Cref{theorem:gkat-checking} is not directly applicable.

A different way to achieve the same effect is to use non-local control flow.
For instance, if GKAT had a way to terminate a \textsf{while} loop without falsifying the loop condition, such as the \textsf{break} statement found in many imperative languages, we could implement our program as follows:
\[
    \begin{array}{l}
    \textsf{if $t$ then}\ \{\ \while{t}{\{\ p;\ \Ifthenelse{t}{\textsf{break}}{p}\ \}}\ \} \\
    \textsf{else}\ \{\ \while{\mathsf{not}\ t}{\{\ p;\ \Ifthenelse{t}{p}{\textsf{break}}\ \}}\ \}
    \end{array}
\]
This approach is also problematic, because the semantics of a \textsf{break}-statement depends on the loop where it occurs, and is therefore not compositional (in a strict sense, following either the language-based or the relational semantics).
Although extending the semantics of GKAT to account for \textsf{break} is doable~\cite{zhang-etal-2024}, such a semantics also violates useful equivalences like
\[
    \forall e \in \GKAT, b \in \BA,
    \quad
    \sem{\while{b}{e}}
        = \sem{\Ifthenelse{b}{\{\ e;\ \while{b}{e}\ \}}{\textsf{assert true}}}
\]
simply because $e$ may contain a \textsf{break} statement, which does not have a meaning outside the loop.

If we want to augment the applicability of \Cref{theorem:gkat-checking} to include programs such as $F$, we are thus left at an impasse: introducing state means making additional assumptions about our programs, while an extension with non-local control flow is not necessarily conservative.

One way to break this impasse could be to add a new control flow construct that mirrors the behavior modeled by $F$, say ``\textsf{repeat $e$ while-changes $b$}''.
The semantics of this operator could be established compositionally using $\sem{e}$ and $\sem{b}$, analogous to how $F$ is defined using $\sigma(p)$ and $\tau(t)$ in \Cref{definition:repeat-while-changes-semantics}.
Such an ``extended GKAT'' would be more expressive, but also begs the question: is the resulting language is expressively complete?
If not, is it possible that \emph{any} extension has some kind of ``blind spot'' in terms of control flow, just like how $F$ is out of reach for GKAT proper?
Our results will imply that
the answers to these questions are \emph{no} and \emph{yes}, respectively.

\section{Kleene Algebra with Tests (KAT)}%
\label{sec:kat}

Up to this point, we have been informal about when control flow is ``programmable''.
Without a more precise definition, we cannot properly express our main result.
We thus set out to formalize this notion, and to that end we must discuss Kleene Algebra with Tests (KAT)~\cite{KAT}.
Unlike GKAT, KAT can express \emph{non-deterministic} control flow, where the next program action is not uniquely determined by the tests.
This generalization allows us to express forms of \emph{deterministic} control flow not covered by GKAT, including the program $F$.
KAT expressions that happen to be deterministic (or \emph{function-preserving}) are therefore a reasonable target in terms of expressivity.

\subsection{Syntax and Intuition}
The syntax of KAT is built on two levels using primitive tests from $T$ and actions from $\Sigma$, just like in GKAT\@.
The lower level is given by $\BA$, as before.
The upper level consists of assertions $b \in \BA$, primitive actions $p \in \Sigma$, and several operators that resemble those of regular expressions:
\[
    \KAT \ni e, f ::= b \in \BA \mid p \in \Sigma \mid e + f \mid e \cdot f \mid e^*
\]
Intuitively, the program $e + f$ should be read as ``choose non-deterministically between running $e$ or $f$'', while $e^*$ is a program that runs $e$ some non-deterministic (possibly zero) number of times.
The program $e \cdot f$ (we usually drop $\cdot$, writing $ef$) denotes the sequential composition of $e$ and $f$.

Based on these intuitions one can make the case that KAT captures deterministic program compositions~\cite{KAT}.
For instance, $\Ifthenelse{b}{e}{f}$ can be modeled by $b e + \overline{b} f$, which non-deterministically chooses between asserting $b$ and executing $e$, or asserting the opposite and running $f$; since only one of these branches ``survives'', the program is deterministic.
Similarly, $\while{b}{e}$ can be modeled by ${(b e)}^* \overline{b}$, which repeats $e$ some number of times, but only keeps the runs where (1)~$b$ is true before every instance of $e$, and (2)~$b$ is false when the repetition halts.

\subsection{Relational Semantics}%
\label{sec:relational-semantics-kat}

We can develop a denotational semantics to KAT programs, just like we did for GKAT\@.
The only important difference is that, this time around, we need to account for nondeterminism by using \emph{relations} on the set of states $S$ instead of partial functions.
To further reflect this, we can also generalize the interpretation $\sigma$ to return a relation on $S$, instead of a partial function.

\begin{definition}[\cite{kozen-smith-1996}]
Suppose we have a valuation of tests $\tau: T \to 2^S$, and for each $p \in \Sigma$ a relation $\sigma(p) \subseteq S \times S$ that represents the (possibly non-deterministic) behavior of $p$.
The semantics of $e \in \KAT$ under this interpretation is a relation $\mathcal{R}\sem{e}_\tau^\sigma$, defined inductively:
\begin{align*}
    \mathcal{R}\sem{b}_\tau^\sigma &= \{ \langle s, s \rangle : s \in \sem{b}_\tau \}
        & \mathcal{R}\sem{e + f}_\tau^\sigma &= \mathcal{R}\sem{e}_\tau^\sigma \cup \mathcal{R}\sem{f}_\tau^\sigma
        & \mathcal{R}\sem{e^*}_\tau^\sigma &= (\mathcal{R}\sem{e}_\tau^\sigma)^* \\
    \mathcal{R}\sem{p}_\tau^\sigma &= \sigma(p)
        & \mathcal{R}\sem{e f}_\tau^\sigma &= \mathcal{R}\sem{e}_\tau^\sigma \circ \mathcal{R}\sem{f}_\tau^\sigma
\end{align*}
In the above, we write $R^*$ to denote the reflexive-transitive closure of a relation $R$ on $S$.

As before, we can choose to think of this semantics as a dependently typed function:
\[
    \mathcal{R}\sem{-}: \KAT \to \forall S : \mathsf{Set}, (\underbrace{T \to 2^S}_{\text{test interp. $\tau$}}) \to  (\underbrace{\Sigma \to 2^{S \times S}}_{\text{action interp. $\sigma$}}) \to 2^{S \times S}
\]
\end{definition}

\begin{example}\label{ex:relational-semantics}
Consider the test alphabet $T = \{t\}$ and the action alphabet $\Sigma = \{p,q\}$, and let $S = \{s_1,s_2,s_3\}$, with $\tau: T \to 2^S$ such that $\tau(t) = \{s_1\}$, and $\sigma: \Sigma \to 2^{S \times S}$ such that $\sigma(p) = \{\langle s_1,s_2\rangle\}$ and $\sigma(q) = \{\langle s_1,s_3\rangle,\langle s_2,s_3\rangle\}$.
Then $\sem{t \cdot q}_\tau^\sigma = \{\langle s_1,s_3\rangle\}$ and $\sem{p + q}_\tau^\sigma = \{\langle s_1,s_2\rangle, \langle s_1, s_3\rangle, \langle s_2,s_3\rangle\}$.
\end{example}

The embedding of GKAT constructs within KAT is sound w.r.t.\ this semantics --- e.g., for any $t \in T$ and $p, q \in \Sigma$, we have that $\mathcal{R}\sem{t p + \overline{t} q}_\tau^\sigma = \sem{\Ifthenelse{t}{p}{q}}_\tau^\sigma$,\footnote{Here, and throughout this paper, we will treat partial functions as (functional) relations whenever convenient.} provided $\sigma(p)$ and $\sigma(q)$ are partial functions.
Thus, the relation on the left is functional, as long as $\sigma(p)$ and $\sigma(q)$ are.
The same holds for the embedding of $\while{t}{p}$ as $(t p)^* \overline{t}$.
This means that KAT can express the same kind of deterministic control flow as GKAT\@.
However, KAT is more powerful, because it can also express the program $F$ from before.
Specifically, the following holds when $\sigma(p)$ is functional:
\[
    \mathcal{R}\sem{t p (\overline{t} p t p)^* (t + \overline{t} p \overline{t}) + \overline{t} p (t p \overline{t} p)^* (\overline{t} + t p t)}_\tau^\sigma = F_\tau^\sigma
\]
We can dissect the expression above a little bit more to gain a better understanding of why this is true.
The non-deterministic branch $t p (\overline{t} p t p)^* (t + \overline{t} p \overline{t})$ represents the possibility that $t$ starts out being true, after which $p$ is run ($tp$).
At that point, a successful run may bounce back and forth between $t$ being false and then true as $p$ is run some number of times ($(\overline{t}ptp)^*$).
Finally, the program finishes either because $t$ is true, or because $t$ is false and remains so after another instance of $p$ ($t + \overline{t}p\overline{t}$).
The other half of the program covers the case where $t$ starts out false.

In the next section, we will make the case that the operators on partial functions expressible as KAT programs $\mathcal{R}\sem{e}$ that preserve partial functions are a reasonable notion of programmable control flow.
Before we can do that, we must first discuss an alternative semantics of KAT\@.

\subsection{Language Model}

The relational (resp.\ partial function) semantics of KAT (resp.\ GKAT) is parameterized by an interpretation of the primitive symbols.
This is useful when studying concrete cases, but it may also complicate arguments about interpretations in general, which is usually what we want to do.
To this end, it is more convenient to consider a \emph{language model}~\cite{KAT}, which turns out to be closely related to the parameterized model.
In the language model, a (G)KAT expression given a semantics in terms of \emph{guarded strings}, which represent successful (symbolic) executions of a program, recording which properties were true at which point, and the actions executed in between.\looseness=-1 

We write $\At_T$ for the atoms of the free Boolean algebra generated by $T$ --- i.e., $\At_T = 2^T$.
We will drop the subscript when $T$ is clear from context or immaterial.
One can think of the elements of $T$ as Boolean variables and each $\alpha \in \At_T$ as a truth assignment, where $t \in \alpha$ precisely when $t$ is true.

\begin{definition}[Guarded languages~\cite{cohen-kozen-smith-1996}]%
A \emph{guarded string} (over test alphabet $T$ and action alphabet $\Sigma$) is an element of the regular language $\At_T \cdot (\Sigma \cdot \At_T)^*$.
A \emph{guarded language} is a set of guarded strings.
We will denote  the set of all  regular languages of guarded strings over the test alphabet $T$ and action alphabet $\Sigma$ by $\G(\Sigma,T)$, or simply $\G$ when $\Sigma$ and $T$ do not matter.
\end{definition}

Intuitively, a guarded string like $\alpha p \beta q \gamma$ represents a program execution that starts in a situation where (only) the tests in $\alpha$ are true, and leads to a point where the tests in $\gamma$ are true.
Along the way, $p$ is executed; this takes us to a point where the tests in $\beta$ hold, after which $q$ is run.

To form the set of guarded strings representing all possible executions of a program, we need a way to compose them while respecting their internal structure.
This is accomplished as follows.

\begin{definition}[Guarded composition~\cite{cohen-kozen-smith-1996}]%
\label{def:guarded-comp}
If $w$ is a guarded string of the form $w'\alpha$ and $x$ is a guarded string of the form $\alpha x'$ (in other words, if the last symbol of $w$ coincides with the first symbol of $x$), then the \emph{guarded composition} of $w$ and $x$, denoted $w\diamond x$, is the guarded string $w' \alpha x'$.
This operation naturally extends to guarded languages $L$ and $K$, as follows:
\[
    L \diamond K =
        \{ w \diamond x: w \in L, x \in K \text{ such that $w \diamond x$ is well defined}\}
\]
We define for $n \in \mathbb{N}$ the repeated guarded concatenation $L^{(n)}$, and the \emph{guarded star} $L^{(*)}$ by
\begin{mathpar}
    L^{(0)} = \At_T
    \and
    L^{(n+1)} = L \diamond L^{(n)}
    \and
    L^{(*)} =
        \bigcup\{L^{(n)} \mid n \in \mathbb{N}\}
\end{mathpar}
\end{definition}

The idea behind guarded composition is as follows: if $w \diamond x$ is defined, then $x$ represents an execution that continues in the same context (i.e., with the same tests being true) as where $w$ left off.
The guarded composition of guarded languages lets guarded strings on the right extend any compatible guarded string on the left to form a new trace.
The guarded star then generalizes this to unbounded repetition, like how the Kleene star extends language concatenation.
These are exactly the operators we need to define the language semantics of KAT, as follows.

\begin{definition}[Language semantics of KAT~\cite{cohen-kozen-smith-1996}]
Let $\alpha \in \At_T$ and $b \in \BA$.
We write $\alpha \leq b$ when $b$ is true under the truth assignment encoded in $\alpha$.
More formally, we define $\sem{-}_{\At}: \BA \to 2^{\At_T}$ by induction on tests, as below, and write $\alpha \leq b$ when $\alpha \in \sem{b}_{\At}$.
\begin{mathpar}
    \sem{\mathsf{false}}_{\At} = \emptyset
    \and
    \sem{\mathsf{true}}_{\At} = \At_T
    \and
    \sem{t}_{\At} = \{ \alpha \in \At_T : t \in \alpha \}
    \\
    \sem{b \vee c}_{\At} = \sem{b}_{\At} \cup \sem{c}_{\At}
    \and
    \sem{b \wedge c}_{\At} = \sem{b}_{\At} \cap \sem{c}_{\At}
    \and
    \sem{\overline{b}}_{\At} = \At_T \setminus \sem{b}_{\At}
\end{mathpar}
Every $e\in \KAT(\Sigma,T)$ now denotes a guarded language $L(e) \in \G(\Sigma,T)$, defined inductively by:
\begin{align*}
    L(b) &= \{ \alpha \in \At_T : \alpha \leq b \}
        & L(e+f) &= L(e) \cup L(f)
        & L(e^*) &= {L(e)}^{(*)} \\
    L(p) &= \{ \alpha p \beta : \alpha, \beta \in \At_T \}
        & L(e \cdot f) &= L(e) \diamond L(f)
\end{align*}
\end{definition}

We hinted earlier that the language semantics of KAT is closely related to its relational semantics.
Having defined the former, we can now formally state the connection, which is that they induce the same kind of equivalences between programs.

\begin{proposition}[\cite{kozen-smith-1996}]%
\label{prop:language-relational-equiv}
For all $e, f \in \KAT$,
    $L(e) = L(f)$ iff
    $\mathcal{R}\sem{e} = \mathcal{R}\sem{f}$.
\end{proposition}
\begin{proof}[Proof sketch]
For the forward implication, take any $\tau: S \to 2^T$ and $\sigma: \Sigma \to 2^{S \times S}$.
By a \emph{labeled path} we will mean a sequence $\pi = s_1 p_1 s_2 p_2 \ldots s_n$
with $s_i\in S$, such that $\langle s_i,s_{i+1}\rangle \in \sigma(p_i)$ for all $i<n$.
Every such labeled path $\pi$ induces a guarded string $w_\pi$, namely $w_\pi = \alpha_1 p_1 \alpha_2 p_2 \ldots \alpha_n$ where $\alpha_i=\{t\mid s_i\in\tau(t)\}$.
It is not difficult to see that, for all KAT-terms $e$, we have
\[
    \mathcal{R}\sem{e}_\tau^\sigma = \{\langle s,s'\rangle\mid \text{there is a labeled path $\pi$ from $s$ to $s'$ such that $w_\pi \in L(e)$}\}
\]
If $L(e) = L(f)$, then the above correspondence implies that $\mathcal{R}\sem{e}_\tau^\sigma = \mathcal{R}\sem{f}_\tau^\sigma$ for all $\tau$ and $\sigma$.

For the backward implication, note that a guarded string $w = \alpha_1 p_1 \alpha_2 p_2 \ldots \alpha_n$ naturally induces an interpretation on $S=\{1, \ldots, n\}$, where $\tau(t) = \{i\in S \mid \alpha_i\leq t\}$ and $\sigma(p) = \{\langle i,i+1\rangle\mid p_i=p\}$.
It is easy to see that, for all KAT expressions $e$, we have $w \in L(e)$ if and only if $\langle 1,n\rangle\in \mathcal{R}\sem{e}_{\tau_w}^{\sigma_w}$.
Given that $\mathcal{R}\sem{e} = \mathcal{R}\sem{f}$, we can then use this property to conclude that $L(e) = L(f)$.
\end{proof}

The correspondence between the language semantics and the relational semantics of KAT allows us to treat an equivalence between KAT terms in either semantics as an equivalence in the other.
This is useful because it prevents us from having to quantify over all possible interpretations of the primitives in proofs of equivalence.
Indeed, decision procedures for KAT equivalence (e.g.\ \cite{pous-2015}) typically operate on the language semantics.
The language semantics also gives rise to an automata model, which forms the basis for our proofs as well as our argument that deterministic KAT expressions are a sensible notion of ``programmable'' control flow (in \Cref{sec:determinism}).

The language semantics of GKAT is derived from that of KAT via the embedding discussed earlier, by treating $\Ifthenelse{b}{e}{f}$ as a shorthand for the KAT expression $be+\overline{b}f$, and $\while{b}{e}$ as a shorthand for the KAT expression $(be)^*\overline{b}$.
Thus every GKAT-term denotes a guarded language.

\begin{definition}[Language semantics of GKAT~\cite{GKAT,KozenTseng}]
For every GKAT term $e \in \GKAT$, we can derive a guarded language $L(e)$ inductively, as follows.
\begin{align*}
    L(b) &= \{ \alpha \in \At : \alpha \leq b \}
        & L(\Ifthenelse{b}{e}{f}) &= (L(b)\diamond L(e))+(L(\overline{b})\diamond L(f)) \\
    L(p) &= \{ \alpha p \beta : \alpha, \beta \in \At \}
        & L(\while{b}{e}) &=  (L(b)\diamond L(e))^*\diamond L(\overline{b}) \\
    L(e; f) &= L(e) \diamond L(f)
\end{align*}
\end{definition}

We will see in Section~\ref{sec:determinism} precisely in which sense GKAT is a ``deterministic'' fragment of KAT\@.

\subsection{Automata Model}

We continue our discussion of KAT by reviewing its automata model~\cite{kozen-2003,chen-pucella-2004,kozen-2017}.
This automaton model will turn out to be crucial for our technical arguments, but also to formalize our notion of ``programmable'' control flow, based on finite state.

We start with a straightforward type of automaton that can accept regular guarded languages; our presentation differs slightly from work in earlier literature, but is essentially equivalent.

\begin{definition}
A \emph{KAT automaton} is a triple $A = (Q, \delta, \iota)$, where $Q$ is a finite set of \emph{states}, $\delta: Q \times \At \to 2^{\{\textsf{accept}\} + \Sigma \times Q}$ is the \emph{transition function} and $\iota: \At \to 2^{\{\textsf{accept}\} + \Sigma \times Q}$ is the \emph{initialization function}.
The guarded language \emph{accepted} in $A$ by $\gamma: \At \to 2^{\{\textsf{accept}\} +\Sigma \times Q}$, denoted $L_A(\gamma)$, is defined as the smallest set satisfying the following rules, for all $\alpha \in \At$:
\begin{mathpar}
    \inferrule{%
        \textsf{accept}  \in \gamma(\alpha)
    }{%
        \alpha \in L_A(\gamma)
    }
    \and
    \inferrule{%
        (p, q') \in \gamma(\alpha) \\
        w \in L_A(\delta(q, -))
    }{%
        \alpha{}pw \in L_A(\gamma)
    }
\end{mathpar}
The language accepted by $A$, denoted $L(A)$, is simply the language of $\iota$ in $A$, i.e., $L_A(\iota)$.
In the sequel, we may lighten notation by simply writing $L_A(q)$ for $L_A(\delta(q, -))$.
\end{definition}

Intuitively, $\textsf{accept} \in \delta(q, \alpha)$ means that $q$ may accept the guarded string $\alpha$, while $(p, q') \in \delta(q, \alpha)$ means that $q$ can transition to $q'$ if it observes the context described in $\alpha$ and performs the action $p$.

\begin{remark}
The type of the transition function for KAT automata is somewhat awkward, and may equivalently be represented by a transition relation ${\rightarrow} \subseteq Q \times \At \times \Sigma \times Q$, accompanied by an acceptance relation ${\downarrow} \subseteq Q \times \At$.
We elect to use this representation because we want to enforce a specific type of determinism in the next section, which will simplify the transition function.

Furthermore, we equip KAT automata with an initialization function $\iota$ rather than an initial state $q_0$ (with its own transition dynamics $\delta(q_0, -)$) because the constructions we develop on (deterministic) KAT automata in sections to follow become a lot easier to specify if we can administer the initial dynamics separately.
This approach mirrors the \emph{pseudostates} used by \citet{GKAT}.
\end{remark}

We can represent a KAT automaton $(Q, \delta)$ visually, by drawing a circle for every state $q \in Q$, and edges leaving $q$ for each element of $\delta(q, \alpha)$, as in \Cref{fig:example-automaton}.
The initialization function is represented by a separate vertex drawn as $\bullet$, and the same transition notation.
We may condense the notation on edges, for instance by writing $b \mid p$ to denote a family of edges, one for each $\alpha \in \At$ with $\alpha \leq b$.

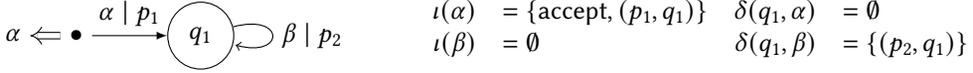
\begin{figure}
    \begin{mathpar}
        \begin{tikzpicture}[baseline]
            \node (init) {$\bullet$};
            \node[state,right=of init] (q1) {$q_1$};
            \draw (init) edge[-latex] node[above] {$\alpha \mid p_1$} (q1);
            \draw (q1) edge[-latex,loop right] node[right] {$\beta \mid p_2$} (q1);
            \node[left=4mm of init] (initacc) {$\alpha$};
            \draw (init) edge[double,double distance=2pt,-implies] (initacc);
        \end{tikzpicture}%
        \and
        \begin{array}{rlrl}
            \iota(\alpha) &= \{ \textsf{accept}, (p_1, q_1) \}
                & \delta(q_1, \alpha) &= \emptyset \\
            \iota(\beta) &= \emptyset
                & \delta(q_1, \beta) &= \{ (p_2, q_1) \}
        \end{array}
    \end{mathpar}
    \vspace{-3mm}
    \caption{%
        Visualization of a KAT automaton on $Q = \{ q_1 \}$ (left), with transitions $\delta$ and initialization $\iota$ (right).
    }%
    \label{fig:example-automaton}
\end{figure}

\begin{example}
Consider the automaton from \Cref{fig:example-automaton}.
The language accepted by $q_0$ includes the guarded strings $\alpha$ and $\alpha p_1 \beta p_2$, $\alpha p_2 \beta p_2 \beta p_2$, et cetera.
Note how this automaton is non-deterministic, in the sense that reading $\alpha$ in state $q_0$ may either cause it to terminate, or perform $p_1$ and transition to $q_1$.
We will further specify our notion of determinism in the next section.
\end{example}

Given how close KAT expressions (resp.\ KAT automata) are to regular expressions (resp.\ finite automata), it should not be surprising that the languages accepted by KAT automata are precisely those denoted by KAT expressions; this extends the well-known theorem due to \citet{kleene-1956}.

\begin{proposition}%
\label{prop:kat-kleene-thm}
Let $L$ be a guarded language.
Then $L = L(e)$ for some $e \in \KAT$ if and only if $L = L(A)$ for some KAT automaton $A$~\cite{kozen-2003}.
Indeed, for every $e \in \KAT$, there exists a KAT automaton $A$ with a number of states linear in the size of $e$ such that $L(e) = L(A)$.
\end{proposition}

The second part of the above comes with the caveat that the number of \emph{transitions} in $A$ may scale exponentially in the size of $T$.
Consider for instance $t \in T \subseteq \KAT$, whose automaton contains a single state with $2^{|T|-1}$ outgoing transitions, one for each $\alpha \in \At_T$ such that $t \in \alpha$.
To keep the number of states limited, it is also crucial that KAT automata are (by default) non-deterministic; this allows the generalization of an efficient method to convert regular expressions to non-deterministic finite automata, such as the approach proposed by \citet{antimirov-1996}.

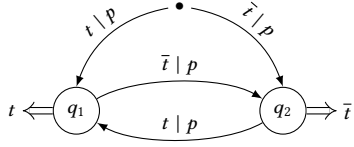
\begin{figure}\scriptsize
   \vspace{-2mm}
    \begin{tikzpicture}
        \node (init) {$\bullet$};
        \node[state,below left=of init] (q1) {$q_1$};
        \node[state,below right=of init] (q2) {$q_2$};
        \draw (init) edge[-latex,bend right] node[above,sloped] {$t \mid p$} (q1.north);
        \draw (init) edge[-latex,bend left] node[above,sloped] {$\overline{t} \mid p$} (q2.north);
        \draw (q1) edge[-latex,bend left,looseness=0.7] node[above] {$\overline{t} \mid p$} (q2);
        \draw (q2) edge[-latex,bend left,looseness=0.7] node[above] {$t \mid p$} (q1);
        \node[left=4mm of q1] (q1acc) {$t$};
        \draw (q1) edge[double,double distance=2pt,-implies] (q1acc);
        \node[right=4mm of q2] (q2acc) {$\overline{t}$};
        \draw (q2) edge[double,double distance=2pt,-implies] (q2acc);
    \end{tikzpicture}
    \vspace{-2mm}
    \caption{KAT automaton modeling the ``repeat $p$ until $t$ changes'' program $F$.}%
    \label{fig:repeat-until-changed-automaton}
\end{figure}

\subsection{Programmable Control Flow}
We started this section with the goal of formalizing what it meant for control flow to be programmable.
To that end, we considered KAT as a generalization of GKAT that includes non-determinism.
We have seen how the relational semantics of KAT extends the partial function semantics of GKAT, and how the former can be related to a model based on languages and automata.
Furthermore, we have seen that KAT can describe control flow that cannot be expressed in GKAT\@.

With these formalisms explained, we can now position KAT as a reasonable model of (non-deterministic) ``programmable'' control flow.
Our primary justification for this is the correspondence between the relational model and KAT automata, whose states can be viewed as different ``decision points'' within a program, with the transitions describing which actions are executed, their preconditions, and the resulting control flow.
If a program has finitely many decision points (which seems reasonable for a finite program), and the ways of moving between those decision points can be described, it should then be possible to derive a KAT automaton for it, and thus a KAT expression.

To further illustrate this point, we note that $F$ can be modeled by the two-state KAT automaton drawn in \Cref{fig:repeat-until-changed-automaton}.
The intended meaning of this automaton is somewhat easier to unravel than the equivalent KAT expression encountered earlier: the initialization function checks whether $t$ holds, and ``remembers'' the outcome by transitioning to $q_1$ or $q_2$; the latter two states accept when, upon arrival, the value of $t$ has not changed, and alternate between each other when it has not.

\section{Determinism}%
\label{sec:determinism}

We have made the case that KAT is a reasonable model of program flow, expressive to the point of being able to describe all finite shapes of control.
The Kleene theorem of KAT, which relates KAT automata to KAT expressions, provides a useful duality: on the one hand, there is the ``flat'' representation in terms of automata, where a program can move between any two states of control, while on the other hand there is the more hierarchical representation in terms of expressions.

We have asserted (but not proved) that this operational-denotational correspondence does not translate to GKAT, because $F$ does not correspond to a GKAT program.
Nevertheless, non-determinism is not necessary to specify $F$: the KAT automaton in \Cref{fig:repeat-until-changed-automaton} is deterministic, in the sense that at most one transition is available in any state at all times; for instance, in $q_1$ the automaton can accept if $t$ holds, or perform $p$ and transition to $q_2$ when it does not.
A reasonable question to ask is whether there is some sort of language that, like GKAT, composes deterministic programs in a deterministic way, and which can express all deterministic control flow described by KAT\@.
The negative answer to this question is (a more formal version of) our main thesis.

In this section, we take the first steps towards that goal, by studying notions of determinism on for the different models of KAT\@.
We start by defining determinism in our setting, as follows.

\begin{definition}
A KAT expression $e \in \KAT(\Sigma,T)$ is \emph{deterministic} when, for all $\tau: T \to 2^S$ and $\sigma: \Sigma \to 2^{S \times S}$, if $\sigma(p)$ is the graph of a partial function over $S$ for all $p\in\Sigma$, then so is $\sem{e}_\tau^\sigma$.
\end{definition}

For instance, Example~\ref{ex:relational-semantics} shows that the KAT expression $p+q$ is not deterministic.
The mere presence of $+$ in a KAT expression does not exclude it from being deterministic, however; for instance, $p + p$ is deterministic, and so is the KAT expression corresponding to the program $F$.

\subsection{Deterministic Languages}
The following is a natural counterpart to determinism under the language semantics.

\begin{definition}[\cite{GKAT}]%
\label{def:deterministic}
A guarded language $L$ is \emph{deterministic} if for all distinct $w, w' \in L$,
(i)~$w$ is not a proper prefix of $w'$, and
(ii)~the first position where they differ is an atom.
\end{definition}
Intuitively, these two conditions mean that (i)~whether an action is executed and (ii)~which action is executed is entirely determined by the preceding string, and by extension, the control point of the program reached so far, as well as the current valuation of the tests.

\begin{example}
The guarded language defined by the KAT-term $p\cdot q$ is deterministic.
Indeed, all strings in $L(p\cdot q)$ have the same length, and they can disagree only on atoms.
The guarded language defined by $p+q$, on the other hand, is not deterministic (assuming $p$ and $q$ are distinct), since $\emptyset p \emptyset$ and $\emptyset q\emptyset$ both belong to $L(p+q)$, which violates the second condition.
The guarded language defined by $p^*$ is not deterministic since it contains both $\emptyset$ and $\emptyset p \emptyset$, which violates the first condition.
\end{example}

In the examples above, the KAT expressions with deterministic languages happen to also be deterministic, and vice versa.
This is no coincidence, as these notions turn out to be the same.

\begin{proposition}%
\label{prop:deterministic-iff-language-deterministic}
Let $e$ be any KAT expression.
Then  $e$ is a deterministic KAT expression if and only if $L(e)$ is a deterministic guarded language.
\end{proposition}
\begin{proof}
We start by developing a standard map from guarded languages to relations.
Let $\tau: T \to 2^S$ and $\sigma: \Sigma \to 2^{S \times S}$.
With each word $w \in (\At \cup \Sigma)^*$, we inductively associate a relation $\sharp(w)$:
\begin{mathpar}
\sharp(\alpha) = \{ \langle s, s \rangle: t \in \alpha \iff s \in \tau(t) \}
\and
\sharp(p) = \sigma(p)
\and
\sharp(wx) = \sharp(w) \circ \sharp(x)
\end{mathpar}
We lift $\sharp$ to languages, writing $\sharp(L)$ for $\bigcup_{w \in L} \sharp(w)$.
A straightforward proof shows $\sharp(L(e)) = \mathcal{R}\sem{e}_\tau^\sigma$.

Now assume that $L(e)$ is a deterministic guarded language, and fix $\tau: T \to 2^S$ and $\sigma: \Sigma \to 2^{S \times S}$ where $\sigma(p)$ is a partial function for every $p \in \Sigma$.
We will show that $\mathcal{R}\sem{e}_\tau^\sigma$ is a partial function.

Let $s_0, s_1, s_2 \in S$ be such that $s_0 \mathrel{\mathcal{R}\sem{e}_\tau^\sigma} s_1$ and $s_0 \mathrel{\mathcal{R}\sem{e}_\tau^\sigma} s_2$.
Because $\sharp(L(e)) = \mathcal{R}\sem{e}_\tau^\sigma$, there exist $w, x \in L(e)$ with $\langle s_0, s_1 \rangle \in \sharp(w)$ and $\langle s_0, s_2 \rangle \in \sharp(x)$.
If $w = x$, then $s_1 = s_2$ because both $\sharp(w)$ and $\sharp(x)$ are partial functions.
We can rule out that $w \neq x$ by observing that, in this case, the conditions of determinism imply that $w = y\alpha w'$ and $x = y\beta x'$ for distinct $\alpha, \beta \in \At$.
Because $\sharp(y)$ is functional, there exists an $s_0' \in S$ such that $\langle s_0,s_0' \rangle \in \sharp(y)$, $\langle s_0',s_1 \rangle \in \sharp(\alpha w')$ and $\langle s_0',s_2 \rangle \in \sharp(\beta x')$.
But then $\langle s_0',s_0' \rangle \in \sharp(\alpha)$ and $\langle s_0',s_0' \rangle \in \sharp(\beta)$, which tells us that $\alpha = \beta$ by definition of $\sharp$, a contradiction.

Next, suppose $L(e)$ is \emph{not} deterministic.
This gives us two cases to dinstinguish:

(1) There are distinct $w, x \in L(e)$ such that the first position where $w$ and $x$ differ is a letter.
    Let $S$ be the set of guarded strings prefixing $w$ or $x$.
    We choose $\tau(t) = \{ w \alpha : w \alpha \in S, t \in \alpha \}$ and $\sigma(p) = \{ (w, wp\alpha) : wp\alpha \in S \}$ (cf.~\cite{kozen-smith-1996}).
    Note that $\sigma(p)$ is functional as a relation on $S$ for every $p \in \Sigma$: given a prefix $y$, if both $yp\alpha$ and $yp\beta$ for $\alpha \neq \beta$ are prefixes of $w$ or $x$, then (w.l.o.g.) $yp\alpha$ is a prefix of $w$ and $yp\beta$ is a prefix of $x$, which contradicts that the first position where they differ is a letter.
    Now apply the definition of $\sharp$ to this $\tau$ and $\sigma$.
    If $\alpha$ is the first atom in both $w$ and $x$ (which exists for the same reason) then $(\alpha,w)\in \sharp(w)$ and $(\alpha,x)\in \sharp(x)$.
    It follows that $(\alpha,w),(\alpha,x)\in\sharp(L(e))$; since $\sharp(L(e)) = \mathcal{R}\sem{e}_\tau^\sigma$, we find that $\mathcal{R}\sem{e}_\tau^\sigma$ is not functional.
    \looseness=-1

 (2)
    There exist $w, x \in L(e)$ with $w$ a proper prefix of $x$.
    Now choose $S$ to be the set of guarded strings prefixing $x$, and $\tau$ and $\sigma$ as before.
    Once again, $\sigma(p)$ is functional for every $p \in \Sigma$: given a prefix $y$ of $x$, if $yp\alpha$ and $yp\beta$ are prefixes of $x$, then $\alpha = \beta$.
    Finally, we turn to $\sharp$ once more, noting that if $w$ (and thus $x$) start with $\alpha$, then $\alpha \mathrel{\sharp(L(e))} w$ and $\alpha \mathrel{\sharp(L(e))} x$, meaning $\mathcal{R}\sem{e}_\tau^\sigma$ is not functional.
\end{proof}

It is worth noting that the interpretation constructed in the above proof is always finite.

\begin{proposition}
Non-determinism of a KAT expression is witnessed by a finite interpretation.
\end{proposition}

GKAT expressions denote deterministic languages.
This can be shown directly, by induction on the syntax of GKAT expressions, but it also follows from \Cref{prop:deterministic-iff-language-deterministic} and the observation that GKAT expressions can be regarded as KAT expressions in a way that is conservative w.r.t.\ their semantics in terms of (functional) relations, pointed out in the previous section.

\begin{theorem}[\cite{GKAT}]\label{thm:gkat-deterministic}
For every GKAT expression $e$, $L(e)$ is deterministic.
\end{theorem}

Nevertheless, GKAT is not capable of expressing \emph{all} deterministic guarded languages.
This was first formally shown by \citet{KozenTseng}; our main result can be viewed as a generalization.

\begin{theorem}[\citet{KozenTseng}]%
\label{thm:KozenTseng}
GKAT is properly contained in deterministic KAT:\@ there exists an $e \in \KAT$ such that $L(e)$ is deterministic, without an $e' \in \GKAT$ with $L(e)=L(e')$.
\end{theorem}

\begin{corollary}
There exists a deterministic $e \in \KAT$ without an $e' \in \GKAT$ with $\mathcal{R}\sem{e} = \sem{e'}$.
\end{corollary}

The expression that witnesses \Cref{thm:KozenTseng} in op.\ cit.\ is rather involved, and so is the accompanying proof.
For the purpose of illustration, we provide a somewhat simpler witness below.

\begin{theorem}[{\cite[Lemma~D.2]{GKATBisim}}]%
\label{thm:gkat-impossible}
There does not exist a GKAT expression whose language semantics coincides with that of the KAT expression $t p (\overline{t} p t p)^* (t + \overline{t} p \overline{t})$.
\end{theorem}

Observe that the expression above appears in the left half of the ``repeat $p$ while $t$ changes'' program from \Cref{sec:relational-semantics-kat}.
In combination with \Cref{prop:language-relational-equiv}, this leads to a proof of \Cref{thm:impossible}.

\begin{proof}[Proof (of \Cref{thm:impossible})]
Suppose that such an expression $e$ exists.
By \Cref{prop:language-relational-equiv} and the fact that the relational semantics of $e$ coincides with that of $f = t p (\overline{t} p t p)^* (t + \overline{t} p \overline{t}) + \overline{t} p (t p \overline{t} p)^* (\overline{t} + t p t)$, we find that the language semantics of $e$ and $f$ must also be the same.
Because both GKAT and KAT expressions are closed under Brzozowski derivatives~\cite{GKAT,chen-pucella-2004}, it should then be possible to obtain from $e$ (resp.\ $f$) a GKAT (resp.\ KAT) expression $e'$ (resp.\ $f'$) that is equivalent to the subexpression $t p (\overline{t} p t p)^* (t + \overline{t} p \overline{t})$.
This contradicts \Cref{thm:gkat-impossible}.
\end{proof}

\subsection{Deterministic KAT Automata}

As we have seen, KAT expressions and KAT automata are equally expressive.
Similarly, deterministic KAT expressions can be characterized by a natural class of KAT automata, namely those $A=(Q,\delta,\iota)$ where $\delta(q, \alpha)$ and $\iota(\alpha)$ always contain at most one element for any $\alpha \in \At$.

\begin{definition}
A \emph{deterministic KAT automaton} is a triple $A = (Q, \delta, \iota)$ where
    $Q$ is a finite set of \emph{states},
    $\delta: Q \times \At \to \{\textsf{accept}, \textsf{reject}\} + \Sigma \times Q$ is the \emph{transition function}, and finally
    $\iota: \At \to \{\textsf{accept}, \textsf{reject}\} + \Sigma \times Q$ is the \emph{initialization} function.
\end{definition}

In other words, a transition of the form
$\delta(q,\alpha)=\textsf{accept}$ can be read as ``when in state $q$, upon reading $\alpha$, accept'', and
$\delta(q,\alpha)=\textsf{reject}$ can be read as ``when in state $q$, upon reading $\alpha$, reject''.

In contrast with KAT automata, $\delta$ and $\iota$ are now functions instead of relations, and their range includes an additional ``$\textsf{reject}$'' value.
Equivalently, of $\delta$ and $\iota$ are partial functions to $\{\textsf{accept}\}+\Sigma\times Q$, with \textsf{reject} being ``undefined''.
Clearly, a KAT automaton $(Q, \delta, \iota)$ where $\delta(q, \alpha)$ and $\iota(\alpha)$ contain at most one element for all $q \in Q$ and $\alpha \in \At$ can be cast as a deterministic KAT automaton.

The semantics of deterministic KAT automata is the same as that of arbitrary KAT automata.
Since we are using a slightly different presentation for $\delta$ and $\iota$ (as functions
instead of relations), we will review the semantics again for the specific case of deterministic KAT automata.

\begin{definition}
For $A = (Q, \delta, \iota)$ a deterministic KAT automaton and $\gamma: \At \to \{\textsf{accept}, \textsf{reject}\} + \Sigma \times Q$, we define the guarded language $L_A(\gamma)$ as the smallest set satisfying the following rules:
\begin{mathpar}
    \inferrule{%
        \gamma(\alpha) =\textsf{accept}
    }{%
        \alpha \in L_A(\gamma)
    }
    \and
    \inferrule{%
        \gamma(\alpha) = (p, q) \\
        w \in L_A(\delta(q, -))
    }{%
        \alpha p w \in L_A(\gamma)
    }
\end{mathpar}
The language of $A$, denoted $L(A)$, is $L_A(\iota)$.
As before, we may write $L_A(q)$ for $L_A(\delta(q, -))$.
\end{definition}

Throughout the sequel, we will use $q \transition{\alpha}{p}_A q'$ as a notation for $\delta(q,\alpha)=(p,q')$.

\begin{proposition}\label{prop:automata-complete}
The languages of guarded strings that are recognized by a deterministic KAT automaton are
precisely the deterministic regular languages of guarded strings.
\end{proposition}
\begin{proof}
For the forward claim, let $A = (Q, \delta, \iota)$ be a deterministic KAT automaton.
Clearly $L(A)$ is regular, as $A$ can be converted to a classic finite automaton that accepts the same language by inserting additional states.
For determinism, it suffices to show that $L_A(\gamma)$ is deterministic for all $\gamma: \At \to \{\textsf{accept}, \textsf{reject}\} + \Sigma \times Q$.
Let $w_1, w_2 \in L_A(\gamma)$ be distinct.
We proceed by induction on $w_1$.
    In the base, $w_1 = \alpha$ for some $\alpha \in \At$.
    Then $\gamma(\alpha) = \textsf{accept}$, which means that $w_1$ cannot be a proper prefix of $w_2$ --- if it were, then $w_2 \in L_A(\gamma)$ being distinct necessitates that $\gamma(\alpha) \neq \textsf{accept}$.
    The atom at the beginning of $w_2$ must then be different from $\alpha$, satisfying the determinism condition.
    For the inductive step, let $w_1 = \alpha p w_1'$; this means that $\gamma(\alpha) = (p, q)$ for some $q \in Q$, with $w_1' \in L_A(\delta(q, -))$.
    If $w_2$ does not start with $\alpha$, then we are done.
    Otherwise, if $w_2$ does start with $\alpha$, then because $w_2 \in L_A(\gamma)$ we know that $w_2 = \alpha p w_2'$ with $w_2' \in L_A(\delta(q, -))$.
    By induction, neither $w_1'$ nor $w_2'$ is a prefix of the other, and the first position where they differ must be an atom, which means the same holds for $w_1$ and $w_2$.\looseness=-1

To see the converse, we define for any $w$ the Brzozowski derivative $L_w = \{ x \mid wx \in L \}$~\cite{brzozowski-1964}, and note that $L_{wp}$ is a guarded language whenever $w$ is a guarded string and $p \in \Sigma$.
We start by claiming that for any deterministic regular language of guarded strings $L$ and any atom $\alpha$, exactly one of the following holds: (1)~$L_\alpha$ is empty, (2)~$L_\alpha = \{ \epsilon \}$, or (3)~there is a unique $p_{L,\alpha} \in \Sigma$ such that $L_{\alpha p_{L,\alpha}}$ is non-empty.
First, note that these are mutually exclusive: if $L_\alpha$ is empty, then $L_{\alpha p}$ is empty for all $p$; if $L_\alpha = \{ \epsilon \}$, then it is not empty; if $L_{\alpha p_{L,\alpha}}$ is non-empty, then $p_{L,\alpha} w' \in L_\alpha$ for some $w'$.
Second, if $L_\alpha$ is neither empty nor $\{ \epsilon \}$, then all strings in $L_\alpha$ start with the same $p_{L,\alpha} \in \Sigma$ by determinism; so one of the cases must apply.
In the third case, $L_{\alpha p_{L,\alpha}}$ must again be deterministic.

Now set $R(L) = \{ L_w : w \in (\At \cdot \Sigma)^* \}$; this set is finite because $L$ is regular, via the Myhill-Nerode theorem. 
We can then define $D: R(L) \times \At \to \{\textsf{accept}, \textsf{reject}\} + \Sigma \times R(L)$ by
\[
    D(L', \alpha) =
        \begin{cases}
        \textsf{reject} & \text{$L'_\alpha$ is empty} \\
        \textsf{accept} & L'_\alpha = \{ \epsilon \} \\
        (p_{L', \alpha}, L'_{\alpha p_{L', \alpha}}) & \text{otherwise}
        \end{cases}
\]
This Brzozowski map is well-defined by the observations above.
Using $D$, we can define the deterministic KAT automaton $(R(L), D, D(L, -))$, which accepts the language $L$ by a straightforward inductive argument on the length of guarded strings.
\end{proof}

\subsection{Deciding Determinism}

While determinism in KAT automata is straightforward to check, finding out whether a given KAT expression is deterministic is another matter altogether.
Fortunately, the characterization of deterministic KAT expressions in terms of deterministic guarded languages and the Kleene theorem for KAT together provide us a way to analyze the complexity of this question.

\begin{theorem}
The following problem is coNP-complete: \emph{given a KAT expression $e$, is $L(e)$ deterministic?}
The same problem is solvable in polynomial time if $|T|$ is bounded by a constant.
\end{theorem}

\begin{proof}
We will first prove the second part of the statement, namely that the problem is solvable in polynomial time for fixed $k=|T|$ (where $T$ is the test alphabet).
We translate $e$ to a (non-deterministic) KAT automaton $A = (Q,\delta,\iota)$ with $L(A) = L(e)$ in polynomial time, with a number of states polynomial in the size of $e$, per \Cref{prop:kat-kleene-thm}.
We further prune the state space of $A$ to ensure that all of its states are \emph{live}~\cite{GKAT}, i.e., they can reach a state $q$ such that $\delta(q, \alpha)$ for some $\alpha \in \At$; this too can be done in polynomial time (provided $|T|$ is fixed).

Let $\gamma, \eta: \At \to 2^{\{ \textsf{accept} \} + \Sigma \times Q}$.
We say that $(\gamma, \eta)$ is \emph{divergent} if for some $\alpha \in \At$:
\begin{enumerate}
    \item
    $\textsf{accept} \in \gamma(\alpha)$ while $\eta(\alpha)$ includes a pair of the form $(p, q)$ (or vice versa); or
    \item
    $\gamma(\alpha)$ (resp. $\eta(\alpha)$) includes a pair $(p,q)$ (resp.\ $(p', q')$) with $p \neq p'$; or
    \item
    $\gamma(\alpha)$ (resp. $\eta(\alpha)$) includes a pair $(p,q)$ (resp.\ $(p, q')$) with $(\delta(q, -), \delta(q', -))$ divergent.
\end{enumerate}

\medskip\par\noindent\emph{Claim 1: }
$L(e)$ is deterministic if and only if $(\iota, \iota)$ is not divergent.

\medskip\par\noindent\emph{Proof of claim: }
An inductive proof shows that $(\gamma, \eta)$ is divergent if and only if there exist distinct $w \in L_A(\gamma)$ and $x \in L_A(\eta)$ where $w$ is a prefix of $x$ (or vice versa), or the first position where $w$ and $x$ differ is a letter.
Thus $L(e) = L(A) = L_A(\iota)$ is deterministic if and only if $(\iota, \iota)$ is not divergent.

\medskip\par\noindent
Standard arguments show that for all $q, q' \in Q$, we can test whether $(\delta(q, -), \delta(q', \alpha))$ is divergent within polynomial time (provided $|T|$ is fixed) and by extension whether $(\iota, \iota)$ is divergent.

Next we proceed to the first part of the statement.
For the coNP upper bound, we argue as follows: for a KAT expression $e$, by a \emph{counterexample for determinism} we will mean a pair of guarded strings $w,w'\in L(e)$ such that either
(i)~$w$ is a proper prefix of $w'$, or
(ii)~the first position where $w$ and $w'$ differ is an element of the action alphabet $\Sigma$.
By definition, $L(e)$ is deterministic if and only if there is no counterexample.
The following two claims together show that the problem is in coNP\@.

\medskip\par\noindent\emph{Claim 2: }
if $L(e)$ is \emph{not} deterministic, there exists a counterexample of size polynomial in that of $e$.

\medskip\par\noindent\emph{Proof of claim: }
Standard automata-theoretic arguments show that if $(\gamma, \eta)$ is divergent, then the $w \in L_A(\gamma)$ and $x \in L_A(\eta)$ from the proof of the previous claim are of size linear in the number of states in $A_e$.
Since $L(e)$ is not deterministic, by Claim 1, $(\iota, \iota)$ is divergent; the $w, w' \in L_{A_e}(\iota) = L(A_e) = L(e)$ obtained this way are the counterexample to determinism.
This concludes the proof of the claim, which we stress concerns only the \emph{existence} of a polynomial-sized counterexample.

\medskip\par\noindent\emph{Claim 3: } it can be decided in polynomial time whether a given pair of guarded strings $(w,w')$ is a counterexample for a given KAT expression $e$.

\medskip\par\noindent\emph{Proof of claim: }
It suffices to show that it can be decided in polynomial time (combined complexity) whether a given guarded string $w$ belongs to $L(e)$ for a given KAT expression $e$.
This follows from known results for Propositional Dynamic Logic (PDL), namely the fact that the combined complexity of model checking for PDL is in polynomial time~\cite{FisherLadner79}, together with the fact that KAT is a fragment of PDL under the relational semantics~\cite{KAT} (recall that guarded strings can be viewed as a special case of relational interpretations, cf.\ the proof of~\Cref{prop:language-relational-equiv}).
This can also easily be argued directly, using a standard ``bottom-up'' evaluation procedure: for $i,j\leq |w|$, let $w_{i,j}$ denote the infix of the guarded string $w$ starting at position $i$ and ending at position $j$.
Given a KAT expression $e$ and a guarded string $w$, we can then inductively compute for each subexpression $e'$ of $e$ the set of all pairs $(i,j)$ for which it holds that $w_{i,j}\in L(e')$.
Because $e$ is a subexpression of itself, this tells us whether $w$ belongs to $L(e)$.

\medskip\par\noindent
Claim 2 and 3 together immediately imply that the problem is in coNP\@.
Hardness can be shown by a reduction from the propositional \emph{satisfiability} problem: for propositional formula $\phi$, let $e_\phi$ be the KAT expression $b_\phi\cdot(p_1 + p_2)$, where $b_\phi \in \BA$ is the corresponding test (in which each proposition letter is now treated as an atomic test) and where $p_1, p_2 \in \Sigma$ are distinct actions.
It is easy to see that $\phi$ is satisfiable if and only if $e_\phi$ is \emph{not} deterministic.
\end{proof}

The upshot of this result is that it is fixed parameter tractable to verify whether a KAT expression intended to model deterministic control flow is indeed deterministic.
However, this does \emph{not} imply that equivalence of deterministic KAT expressions is possible with similar complexity; indeed, equivalence of (general) KAT expressions is PSPACE-complete~\cite{cohen-kozen-smith-1996}.
Part of the added value of a syntax for (a subset of) deterministic KAT expressions such as GKAT is that it provides a method to obtain a deterministic KAT automata for an expression that is linear in the size of that expression, makes equivalence checking tractable~\cite{GKAT,zhang-etal-2024}.

\section{Regular Control Flow Operations}%
\label{sec:regular-operations}

In the previous section, we showed that determinism is a robust concept within KAT, which can be described using the relational semantics, the language semantics, and KAT automata.
We also saw that GKAT cannot express all deterministic KAT operations.
In this section, we lay the groundwork for generalizing GKAT using \emph{regular control flow operations}.
Intuitively, these are operations that arise from a deterministic KAT expression containing action variables $p_1, \ldots, p_n$ and test variables $t_1, \ldots, t_m$.
We encountered these kinds of operations before, when we encoded $\Ifthenelse{t}{p_1}{p_2}$ in KAT using the deterministic expression $t p_1 + \overline{t} p_2$ (cf.\ Example~\ref{ex:operations} below).

To make this precise, we need the notion of a \emph{KAT-morphism}.
Recall that $\KAT(\Sigma,T)$ denotes the set of KAT-expressions generated by $\Sigma$ and $T$, and that $\BA(T)$ denotes the set of tests over $T$.

\begin{definition}[KAT-morphism]
By a  \emph{KAT-morphism} $(\ssigma,\ttau): (\Sigma_0,T_0) \to (\Sigma_1,T_1)$, we will mean a pair $(\ssigma,\ttau)$ of functions where $\ssigma: \Sigma_0 \to \KAT(\Sigma_1,T_1)$
and $\ttau:T_0\to \BA(T_1)$.
This extends naturally to a corresponding map $\apply^\ssigma_\ttau$ from $\KAT(\Sigma_0,T_0)$ to $\KAT(\Sigma_1,T_1)$: for $e \in \KAT(\Sigma_0,T_0)$, $\apply^\ssigma_\ttau(e)$ is the result of replacing every test $t \in T_0$ in $e$ by $\ttau(t)$ and replacing every action $p \in \Sigma_0$ in $e$ by $\ssigma(p)$.
\end{definition}

We can think of $\apply^\ssigma_\ttau$ as performing syntactic substitutions specified by $\ssigma$ and $\ttau$ in the obvious way.
There is also a corresponding substitution operation at the level of guarded languages.

\begin{definition}[Guarded language morphism]%
\label{def:apply-lang}
A \emph{guarded language morphism} is a pair $(\ssigma,\ttau)$ of functions $\ssigma: \Sigma_0 \to \G(\Sigma_1,T_1)$ and $\ttau:T_0\to \BA(T_1)$.
Let $(\ssigma, \ttau)$ be such a guarded language morphism.

An atom $\beta \in \At_{T_1}$ is \emph{$\ttau$-consistent} with $\alpha \in \At_{T_0}$ when for all $t \in T_0$, $\alpha \leq t$ iff $\beta \leq \ttau(t)$.
Given $L \in \G(\Sigma_0,T_0)$, we write $\apply_{\ttau}(L)$ for the guarded language in $\G(\Sigma_0, T_1)$ consisting of guarded strings $\beta_0 p_0 \beta_1 \cdots \beta_{n-1} p_{n-1} \beta_n$ where $\alpha_0 p_0 \alpha_1 \cdots \alpha_{n-1} p_{n-1} \alpha_n \in L$ and each $\alpha_i$ is $\ttau$-consistent with $\beta_i$.

Given $L \in \G(\Sigma_0, T_1)$, the guarded language in $\G(\Sigma_1, T_1)$ consisting of guarded concatenations (cf. \Cref{def:guarded-comp}) $\alpha_0 \diamond w_0 \diamond \alpha_1 \diamond \cdots \diamond \alpha_{n-1} \diamond w_{n-1} \diamond \alpha_n$, where $\alpha_0 p_0 \alpha_1 \cdots \alpha_{n-1} p_{n-1} \alpha_n \in L$ with $w_i \in \ssigma(p_i)$ is denoted $\apply^\ssigma(L)$.
We write $\apply^\ssigma_\ttau$ for $\apply^\ssigma \circ \apply_\ttau: \G(\Sigma_0, T_0) \to \G(\Sigma_1, T_1)$.
\end{definition}

Note: the order of application matters in the above definition: $\apply^\ssigma$ must be performed after $\apply_\ttau$.
The following propositions then follow immediately from the definitions:

\begin{proposition}\label{prop:apply-well-defined}
For all KAT-morphisms $(\ssigma,\ttau): (\Sigma_0,T_0)\to (\Sigma_1,T_1)$ and $e \in KAT(\Sigma_0,T_0)$, we have that $ L(\apply^\ssigma_\ttau(e)) = \apply^{\ssigma'}_\ttau(L(e))$
where $\ssigma'$ is given by $\ssigma'(p) = L(\ssigma(p))$.
\end{proposition}

\begin{proposition}
Let $L\in \G(\Sigma_0,T_0)$ be a deterministic guarded language and $(\ssigma,\ttau):(\Sigma_0,T_0)\to (\Sigma_1,T_1)$ a guarded language morphism.
If $\ssigma(p)$ is deterministic for all $p\in\Sigma_0$, then so is $\apply^\ssigma_\ttau(L)$.
\end{proposition}

Having this in place, we can now talk about \emph{operations defined by KAT terms}.

\begin{definition}
For every KAT term $e\in \KAT(\{p_1, \ldots, p_n\},\{t_1, \ldots, t_m\})$, we define $O_e$ to be the $n+m$-ary function that takes as input
$n$ guarded languages $L_1, \ldots, L_n\in \G(\Sigma,T)$ and
$m$ Boolean expressions $b_1, \ldots, b_m\in \BA(T)$, and
outputs the guarded language $\apply^\ssigma_\ttau(L(e))$, where $\ssigma$ and $\ttau$ are given by $\ssigma(p_i) = L_i$ and $\ttau(t_i) = b_i$.
We refer to such functions on guarded languages as \emph{regular operations}.
If $e$ is deterministic, we call $O_e$ a \emph{(deterministic) regular control flow operation}.
\end{definition}

\begin{example}\label{ex:operations}
Examples of regular control flow operations include
\begin{itemize}
    \item
    if-then-else, defined by the deterministic KAT expression $t_1 p_1 + \overline{t_1}p_2$;
    \item
    sequential composition, defined by the deterministic KAT expression $p_1 p_2$;
    \item
    while-do, defined by the deterministic KAT expression $(t_1 p_1)^*\overline{t_1}$; and
    \item
    repeat-while-changes, defined by the deterministic KAT expression from \Cref{sec:relational-semantics-kat}.
\end{itemize}
\end{example}

Clearly, the guarded languages denoted by GKAT terms can be formed using the first three operations above, starting from the languages $L(p)$ for $p \in \Sigma$, and $L(t)$ for $t \in \BA$.
\Cref{thm:KozenTseng} says that these languages are a proper subset of all deterministic KAT languages.

Interestingly, intersection of guarded languages is \emph{not} a regular control flow operation, even though it is a binary operation on guarded languages that preserves determinism.
\begin{lemma}%
\label{lem:intersection-nonregular}
There is no $e \in \KAT(\{p_1, p_2\}, \emptyset)$ such that $O_e(L_1, L_2) = L_1 \cap L_2$ for all $L_1, L_2 \in \G$.
\end{lemma}
\begin{proof}
Suppose towards a contradiction that such an $e$ does exist.
Take $\Sigma = \{ u_1, u_2, u_3 \}$ and $T = \emptyset$, and choose $L_1 = \{ \emptyset u_1 \emptyset, \emptyset u_2 \emptyset \}$ and $L_2 = \{ \emptyset u_2 \emptyset, \emptyset u_3 \emptyset \}$.
We can then derive that
\begin{mathpar}
    \emptyset u_2 \emptyset \in L_1 \cap L_2 = O_e(L_1, L_2) = \apply^\ssigma_\ttau(L(e))
    \and
\end{mathpar}
where $\ssigma(p_i) = L_i$ and $\ttau$ is the unique map from $\emptyset$ to $\BA(\emptyset)$.
By definition of $\apply^\ssigma_\ttau$, there must then exist a $w = \emptyset p'_0 \emptyset \cdots \emptyset p'_{n-1} \emptyset \in L(e)$ such that $\emptyset u_2 \emptyset$ is of the form $\emptyset \diamond w_0 \diamond \emptyset \diamond \cdots \diamond \emptyset \diamond w_{n-1} \diamond \emptyset$ where if $p'_i = p_j$ then $w_i \in L_j$.
But then there is precisely one $w_i$ such that $w_i = \emptyset u_2 \emptyset$, while $w_j = \emptyset$ for all $i \neq j$.
If $p'_i = p_1$, then using this $w$ we can also derive that $\emptyset u_1 \emptyset \in O_e(L_1, L_2)$, while if $p'_i = p_2$ then $\emptyset u_3 \emptyset \in O_e(L_1, L_2)$; either possibility contradicts that $O_e(L_1, L_2) = L_1 \cap L_2$, so $e$ cannot exist.
\end{proof}

\section{Extensions of GKAT with Regular Control Flow Operations}%
\label{sec:main}

We now consider sets of languages that can be built using finitely many regular control flow operations, for example by adding repeat-while-changes to GKAT\@.
More precisely, if $\mathcal{O}$ is a set of regular control flow operations, then we say $\mathcal{O}$ \emph{generates} $L \in \G(\Sigma, T)$ if $L$ can be obtained by composing $L(p)$ for $p \in \Sigma$ and $L(b)$ for $b \in \BA(T)$ using (possibly multiple applications of) operations from $\mathcal{O}$.
$\mathcal{O}$ generates a \emph{set} of guarded languages if it generates each member.
Our main theorem implies that no finite set of regular control flow operators (and hence no finite extension of GKAT) can generate all deterministic guarded languages.
Note that the infinitary extension of GKAT with \emph{all} regular control flow operations is expressively complete by construction.

Our arguments hinge on the following family of languages, which can be thought of as a generalization of the languages put forward by \citet{KozenTseng,GKATBisim}.

\begin{definition}\label{def:Lk}
For $k>0$, the language $L_k$ is defined as follows: let $T=\{t_1, \ldots, t_k\}$ and $\Sigma = \{p_1, \ldots, p_k\}$.
For $i\leq k$, $\alpha_i$ is the singleton set $\{t_i\}$.
$L_k$ consists of all guarded strings of the form
\[\alpha_{i_1} ~ p_{1} ~ \alpha_{i_2} ~ p_{i_1} ~ \alpha_{i_3}  ~ p_{i_2} \ldots \alpha_{i_n} ~ p_{i_{n-1}} ~  \alpha_{i_n}  \]
with $n\geq 1$ where $i_1\neq 1$ and each $i_j\in \{1, \ldots, k\}$ such that $i_j \neq i_{j+1}$ for all $j<n$.
\end{definition}

It is clear that $L_k$ is a deterministic regular guarded language (cf.\ \Cref{def:deterministic}).
For any finite set of regular control flow operators $\mathcal{O}$ and sufficiently large $k$, $L_k$ cannot be generated by $\mathcal{O}$.

\begin{theorem}\label{thm:GKATO-main}
Let $\mathcal{O}$ be any finite set of regular control flow operations (possibly including the operations of GKAT).
Then there exists a $k > 0$ such that $L_k$ is not generated by $\mathcal{O}$.
\end{theorem}

\begin{corollary}\label{cor:main}
The deterministic fragment of KAT is not  generated by a finite set of regular control flow operations.
\end{corollary}

\begin{remark}
Another way to phrase \Cref{cor:main} is as follows: there is no finite set of deterministic KAT-expressions from which all other deterministic KAT-expressions can be obtained by closing under substitution. This purely syntactic formulation makes it clear that the result is not inherently tied to the language semantics of KAT, but applies equally well to the relational semantics. As mentioned in the introduction, this also resolves an open question from~\cite{bogaerts2024preservation}.
\end{remark}

The remainder of this section is dedicated to the proof of the above theorem.
Our approach is as follows: by Proposition~\ref{prop:automata-complete}, each regular control flow operation corresponds to a deterministic KAT automaton.
Intuitively, we can view terms built from such operations as recursively composed automata.
We formalize this intuition by defining a suitable composition operation on deterministic KAT automata.
We then consider languages recognized by automata that are recursively composed out of small automata (i.e., with a bounded number of states).
We show that these form a proper subset of the deterministic regular languages of guarded strings
(cf.~\Cref{thm:main} below).

\subsection{Automata Composition}

We will now describe a way to compose deterministic KAT automata (henceforth: automata) out of simpler automata.
The composition operation for KAT automata is along the same lines as for finite state automata.
The main complicating factor is the presence of the atoms guarding transitions, which will require us to address the issue of possible trivial transitions in composed automata.

Let $A = (Q, \delta, \iota)$ be an automaton over program and test alphabet $(\Sigma_0, T_0)$.
Furthermore, let $\ssigma$ assign to each $p \in \Sigma_0$ an automaton $\ssigma(p)$ over $(\Sigma_1, T_1)$, and let $\ttau: T_0 \to \BA(T_1)$.
We will construct an automaton $\compose^\ssigma_\ttau(A)$ over $(\Sigma_1, T_1)$.
This will be the automata-counterpart of the $\apply^\ssigma_\ttau$ operation on languages from Section~\ref{sec:regular-operations}, and it will satisfy the following analogue of \Cref{prop:apply-well-defined}:

\begin{proposition}%
\label{prop:apply-compose}
$L(\compose^\ssigma_\ttau(A)) = \apply^{\ssigma'}_\ttau(L(A))$ where $\ssigma'$ is given by $\ssigma'(p) = L(\ssigma(p))$.
\end{proposition}

In light of~\Cref{def:apply-lang}, we can do this by first developing separate composition operations $\compose^\ssigma$ and $\compose_\ttau$, and then defining $\compose^\ssigma_\ttau = \compose^\ssigma \circ \compose_{\ttau}$.

\subsubsection{Definition of $\compose_\ttau(A)$}
Let $\ttau:T_0\to \BA(T_1)$ and let $A$ be an automaton over $(\Sigma,T_0)$.
We construct $\compose_{\ttau}(A)$ over $(\Sigma,T_1)$ as follows.
For each $\beta\in\At_{T_1}$ let $\ttau^{-1}(\beta)$ denote the corresponding element of $\At_{T_0}$, that is, the atom where $t \in \alpha$ iff $\beta \leq \ttau(t)$.
Now $\compose_\ttau(A) = (Q,\delta',\iota')$ where $\delta'$ and $\iota'$ are given by $\delta'(q,\beta)=\delta(q,\ttau^{-1}(\beta))$ and
$\iota'(\beta)=\iota(\ttau^{-1}(\beta))$.
It is easy to confirm that:

\begin{proposition}%
\label{prop:apply-compose-1}
$L(\compose_{\ttau}(A)) = \apply_{\ttau}(L(A))$.
\end{proposition}

\subsubsection{Definition of $\compose^\ssigma(A)$}
Let $\ssigma$ be a map from $\Sigma_0$ to automata over $(\Sigma_1,T)$ and let $A$ be an automaton for $(\Sigma_0,T)$.
The basic intuition behind the construction of $\compose^\ssigma(A)$ is simple.
In loose terms, whenever $A$ would execute the action $p$, $\compose^\ssigma(A)$ instead starts executing a trace in $\ssigma(p)$ until it accepts, at which point we proceed as $A$ would.
This way, every occurrence of $p$ in a guarded string accepted by $A$ is replaced with a guarded string accepted by $\ssigma(p)$.

This is fairly easy to define as a construction on automata, except for one slippery issue: \emph{what if the automaton $\ssigma(p)$ accepts immediately?}
In such a case, we should consult the next state of $A$, which may execute another action $p'$ and transition to $q''$, meaning we need to consult $\ssigma(p')$, and so on.
In general, this may trigger an infinite series of actions $p, p', p'', \ldots$ where each corresponding automaton accepts immediately --- in which case we want the composition to reject, as no useful progress is being made.
To make this idea more precise, the following definition is helpful.

\begin{definition}%
\label{def:delta-hat}
We define $\hat{\delta}: Q \times \At_{T} \to \{\textsf{accept}, \textsf{reject}\} + \Sigma_0 \times Q$ as follows:
\[\hat{\delta}(q, \alpha) =
  \begin{cases}
  \mathsf{accept} & \text{whenever } \delta(q, \alpha) = \mathsf{accept} \\
  (p,q') &\text{whenever } \delta(q,\alpha) = (p, q')  \text{ and } \iota_p(\alpha) \neq \textsf{accept} \\
  \hat{\delta}(q',\alpha) &\text{whenever } \delta(q, \alpha) = (p, q') \text{ and }  \iota_p(\alpha) =\textsf{accept} \\
   \textsf{reject} & \text{otherwise (i.e., if no value is derivable by the above clauses)}
\end{cases}\]
\end{definition}

This definition is recursive, but the last case handles infinite regression, which makes $\hat{\delta}$ well-defined.
We could alternatively define it as the least fixed point of a Scott-continuous function in a DCPO (cf. \Cref{remark:ill-founded} and \Cref{definition:repeat-while-changes-semantics}), but we chose to keep this version for the sake of clarity.

Intuitively, $\hat{\delta}$ ``shortcuts'' the definition of $\delta$ based on the atoms accepted by $\ssigma(p)$ for $p\in\Sigma_0$: when $\delta(q, \alpha) = (p, q')$ and $\iota_p(\alpha) = 1$, it looks at what $q'$ would have done when given $\alpha$.
When this leads to an infinite loop, $\hat{\delta}(q, \alpha)$ rejects.
We can now shortcut the initialization function, too.

\begin{definition}%
\label{def:iota-hat}
We define $\hat{\iota}: \At_{T} \to \{\textsf{accept}, \textsf{reject}\} + \Sigma_0 \times Q$ by
\begin{align*}
    \hat{\iota}(\alpha) &=
        \begin{cases}
        \hat{\delta}(q, \alpha) & \text{if } \iota(\alpha) = (p, q) \text{ and }  \iota_{\ssigma(p)}(\alpha) =\textsf{accept} \\
        \iota(\alpha) & \text{otherwise}
        \end{cases}
\end{align*}
\end{definition}

\begin{figure}
    \begin{mathpar}
        \ssigma(p_1) =
            \begin{tikzpicture}[baseline,every state/.style={scale=0.8}]
                \node (initial) {$\bullet$};
                \node[state,right=15mm of initial] (r1) {$r_1$};
                \node[state,right=15mm of r1] (r2) {$r_2$};
                \draw (initial) edge[-latex] node[above] {$t|p$} (r1);
                \draw (r1) edge[-latex] node[above] {$t'|p'$} (r2);
                \node[above=4mm of initial] (initialacc) {$\overline{t}$};
                \draw (initial) edge[double,double distance=2pt,-implies] (initialacc);
                \node[above=4mm of r1] (r1acc) {$\overline{t'}$};
                \draw (r1) edge[double,double distance=2pt,-implies] (r1acc);
                \node[right=4mm of r2] (r2acc) {$\mathsf{true}$};
                \draw (r2) edge[double,double distance=2pt,-implies] (r2acc);
            \end{tikzpicture}
        \and
        \ssigma(p_2) =
            \begin{tikzpicture}[baseline,every state/.style={scale=0.8}]
                \node (initial) {$\bullet$};
                \node[state,right=15mm of initial] (r3) {$r_3$};
                \node[above=4mm of initial] (initialacc) {$\overline{t''}$};
                \draw (initial) edge[double,double distance=2pt,-implies] (initialacc);
                \node[above=4mm of r3] (r3acc) {$\overline{t''}$};
                \draw (r3) edge[double,double distance=2pt,-implies] (r3acc);
                \draw (initial) edge[-latex] node[above] {$t''|p''$} (r3);
                \draw (r3) edge[-latex,loop right] node[right] {$t''|p''$} (r3);
            \end{tikzpicture}
    \end{mathpar}
    \[
        \begin{tikzpicture}[baseline=0mm,every state/.style={scale=0.8}]
            \node (initial) {$\bullet$};
            \node[state,right=11mm of initial] (s1) {$s_1$};
            \node[state,right=11mm of s1] (s2) {$s_2$};
            \draw (initial) edge[-latex] node[above] {$\mathsf{true}|p_1$} (s1);
            \draw (s1) edge[-latex] node[above] {$\mathsf{true}|p_2$} (s2);
            \node[above=4mm of s2] (s2acc) {$\mathsf{true}$};
            \draw (s2) edge[double,double distance=2pt,-implies] (s2acc);
        \end{tikzpicture}
        \quad\quad\leadsto\quad\quad
        \begin{tikzpicture}[baseline=0mm,every state/.style={scale=0.8}]
            \node (initial) {$\bullet$};
            \node[state,right=11mm of initial] (s1) {$s_1$};
            \node[state,right=11mm of s1] (s2) {$s_2$};
            \draw (initial) edge[-latex] node[above] {$t|p_1$} (s1);
            \draw (initial) edge[-latex,bend right] node[below,sloped] {$\overline{t} \wedge t''|p_1$} (s2);
            \draw (s1) edge[-latex] node[above] {$t''|p_2$} (s2);
            \node[left=4mm of initial] (initialacc) {$\overline{t} \wedge \overline{t''}$};
            \draw (initial) edge[double,double distance=2pt,-implies] (initialacc);
            \node[above=4mm of s1] (s1acc) {$\overline{t''}$};
            \draw (s1) edge[double,double distance=2pt,-implies] (s1acc);
            \node[above=4mm of s2] (s2acc) {$\mathsf{true}$};
            \draw (s2) edge[double,double distance=2pt,-implies] (s2acc);
        \end{tikzpicture}
    \]
    \caption{
        At the top, a map $\ssigma$ from $\Sigma_2 = \{ p_1, p_2 \}$ to automata over $T = \{ t, t', t'' \}$ and $\Sigma = \{ p, p', p'' \}$.
        On the bottom left, an automaton $A = (Q, \delta, \iota)$ over $T_0 = \emptyset$ and $\Sigma_2$, representing the operator $p_1 \cdot p_2$.
        On the bottom right, the automaton $(Q, \hat{\delta}, \hat{\iota})$ obtained from $A$, with $\hat{\delta}$ and $\hat{\iota}$ derived from $\delta$ and $\iota$ per \Cref{def:delta-hat,def:iota-hat}.
    }%
    \label{fig:shortcut-drawing}
\end{figure}
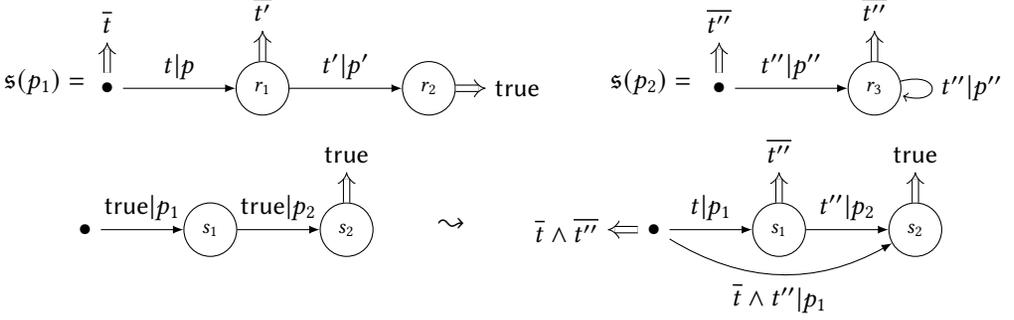

\begin{example}%
\label{ex:shortcut}
Consider the automaton at the bottom left of \Cref{fig:shortcut-drawing}, representing the deterministic KAT operator $p_1 \cdot p_2$.
If we want to compose using $\ssigma$ at the top of the figure, we must first shortcut the $\delta$ and $\iota$ to obtain the automaton at the bottom right.
In this automaton, the initialization map has gained a transition to $s_2$ labeled $\overline{t} \wedge t''|p_1$, because if $\alpha$ is such that $\alpha \leq \overline{t} \wedge t''$, then $\iota(\alpha) = (p_1, s_1)$, and $\iota_{\ssigma(p_1)}(\alpha) = \mathsf{accept}$, so $\hat{\iota}(\alpha)$ takes the value of $\hat{\delta}(s_1, \alpha)$, which is $(p_2, s_2)$ because $\delta(s_1, \alpha) = (p_2, s_2)$ and $\iota_{p_2}(\alpha) \neq \mathsf{accept}$.
The initialization map also gains an accepting transition for all $\alpha \leq \overline{t} \wedge \overline{t''}$.
\end{example}

Now that we have $\hat{\delta}$ and $\hat{\iota}$, we can define our composed automaton.

\begin{definition}%
\label{def:compose-sigma}
Let $A$ be an automaton over $(\Sigma_0, T)$, and let $\ssigma$ map each $p \in \Sigma_0$ to an automaton $A_p = (Q_p, \delta_p, \iota_p)$ over $(\Sigma_1,T)$.
We write $\compose^\ssigma(A)$ for the automaton $(Q', \delta', \iota')$ over $(\Sigma_1,T)$, where \(Q' = \sum_{p \in \Sigma_0} Q_p \times Q\) and
$\delta', \iota'$ are as follows (for $q_p \in Q_p$, $q_{p'} \in Q_{p'}$, $q \in Q$ and $\alpha \in \At_{T}$):
\begin{align*}
    \delta'((q_p, q), \alpha) &=
        \begin{cases}
        (p'', (q_p', q)) & \text{if } \delta_p(q_p, \alpha) = (p'', q_p') \\
        (p'', (q_{p'}, q')) & \text{if } \delta_p(q_p, \alpha) =\textsf{accept}, \hat{\delta}(q, \alpha) = (p', q') \text{ and } \iota_{p'}(\alpha) = (p'', q_{p'}) \\
        \textsf{accept} & \text{if } \delta_p(q_p, \alpha) =\textsf{accept} \text{ and } \hat{\delta}(q, \alpha) =\textsf{accept} \\
        \textsf{reject} & \text{otherwise}
        \end{cases} \\
    \iota'(\alpha) &=
        \begin{cases}
        \textsf{accept} & \text{if }\hat{\iota}(\alpha) =\textsf{accept} \\
        (p'', (q_p, q)) & \text{if } \hat{\iota}(\alpha) = (p', q) \text{ and } \iota_{p'}(\alpha) = (p'', q_p) \\
        \textsf{reject} & \text{otherwise}
        \end{cases}
\end{align*}
\end{definition}

In $\compose^\ssigma(A)$, a state $(q_p, q) \in Q_p \times Q$ represents that we are currently in state $q_p$ of the subautomaton for $\ssigma(p)$, and that execution will continue in state $q$ of $A$ once the subautomaton accepts.
With this in mind, the first case in the definition of $\delta'$ covers transitions within the subautomaton, while the second and third case say what to do when the subautomaton accepts.

\begin{figure}
    \[
        \begin{tikzpicture}[baseline=-12mm,every state/.style={scale=0.8}]
            \node (initial) {$\bullet$};
            \node[draw,dotted,rounded corners=1mm,right=7mm of initial] (ssigmap1') {$\ssigma(p_1)$};
            \node[state,right=7mm of ssigmap1'] (s1) {$s_1$};
            \draw (initial) edge node[above] {$t$} (ssigmap1');
            \draw (ssigmap1') edge[-latex] (s1);
            \node[draw,dotted,rounded corners=1mm,below=4mm of s1] (ssigmap2) {$\ssigma(p_2)$};
            \node[state,below=12mm of s1] (s2) {$s_2$};
            \node[draw,dotted,rounded corners=1mm] at (initial |- s2) (ssigmap1) {$\ssigma(p_1)$};
            \draw (initial) edge node[left] {$\overline{t} \wedge t''$} (ssigmap1);
            \draw (ssigmap1) edge[-latex] (s2);
            \draw (s1) edge node[right] {$t''$} (ssigmap2);
            \draw (ssigmap2) edge[-latex] (s2);
            \node[left=4mm of initial] (initialacc) {$\overline{t} \wedge \overline{t''}$};
            \draw (initial) edge[double,double distance=2pt,-implies] (initialacc);
            \node[right=4mm of s1] (s1acc) {$\overline{t''}$};
            \draw (s1) edge[double,double distance=2pt,-implies] (s1acc);
            \node[right=4mm of s2] (s2acc) {$\mathsf{true}$};
            \draw (s2) edge[double,double distance=2pt,-implies] (s2acc);
        \end{tikzpicture}
        \;\leadsto\;
        \begin{tikzpicture}[baseline=-12mm,every state/.style={scale=0.8}]
            \node (initial) {$\bullet$};
            \node[state,right=15mm of initial] (r1s1) {$r_1, s_1$};
            \node[state,below=12mm of initial] (r3s2) {$r_3, s_2$};
            \node[state] at (r3s2 -| r1s1) (r2s1) {$r_2, s_1$};
            \node[left=4mm of initial] (initialacc) {$\overline{t} \wedge \overline{t''}$};
            \draw (initial) edge[double,double distance=2pt,-implies] (initialacc);
            \draw (initial) edge[-latex] node[above] {$t|p$} (r1s1);
            \draw (initial) edge[-latex] node[left] {$\overline{t} \wedge t''|p''$} (r3s2);
            \node[right=4mm of r1s1] (r1s1acc) {$\overline{t'} \wedge \overline{t''}$};
            \draw (r1s1) edge[double,double distance=2pt,-implies] (r1s1acc);
            \node[right=4mm of r2s1] (r2s1acc) {$\overline{t''}$};
            \draw (r2s1) edge[double,double distance=2pt,-implies] (r2s1acc);
            \node[below=4mm of r3s2] (r3s2acc) {$\overline{t''}$};
            \draw (r3s2) edge[double,double distance=2pt,-implies] (r3s2acc);
            \draw (r1s1) edge[-latex] node[right] {$t'|p'$} (r2s1);
            \draw (r1s1) edge[-latex] node[above,sloped,pos=0.5] {$\overline{t'} \wedge t''|p''$} (r3s2);
            \draw (r2s1) edge[-latex] node[above,pos=0.4] {$t''|p''$} (r3s2);
            \draw (r3s2) edge[-latex,loop left] node[left] {$t''|p''$} (r3s2);
        \end{tikzpicture}
    \]
    \vspace{-4mm}
    \caption{
        On the left, the automaton $(Q, \hat{\delta}, \hat{\iota})$ from \Cref{fig:shortcut-drawing}, drawn suggestively with the transitions ``entering'' the subautomata or $\ssigma(p_1)$ and $\ssigma(p_2)$, along with the tests that are true when this happens, and the states where the automaton resumes after a subautomaton accepts.
        On the right, the automaton $\compose^\ssigma(A)$, obtained by following the recipe from \Cref{def:compose-sigma}.
    }%
    \label{fig:compose-example}
   \vspace{-3mm}
\end{figure}
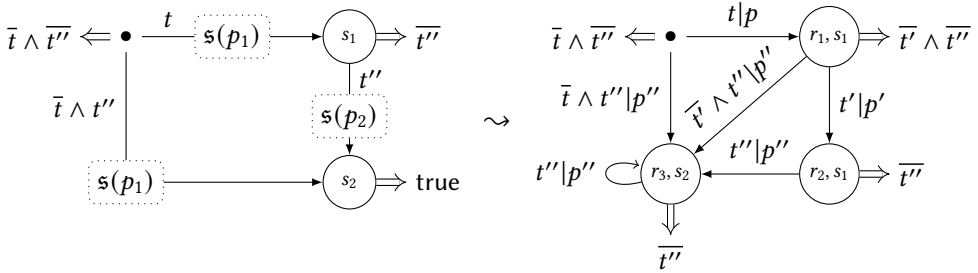

\begin{example}
Returning to the automaton $A$ and the map $\ssigma$ from \Cref{ex:shortcut}, we compute $\compose^\ssigma(A)$ in \Cref{fig:compose-example}.
On the left, we see shortcut version of $A$, but with the automata obtained from $\ssigma$ drawn ``between'' the states.
On the right, these subautomata are expanded to obtain the composed automaton, per \Cref{def:compose-sigma}.
Here, $(r_1, s_1)$ has a transition labeled $\overline{t'} \wedge t' | p''$ to $(r_3, s_2)$, because $(r_1, s_1)$ is in the subautomaton $\ssigma(p_1)$, and when  $\alpha \leq \overline{t'} \wedge t''$, we have $\delta_{p_1}(r_1, \alpha) = \mathsf{accept}$ and $\hat{\delta}(s_1, \alpha) = (p_2, s_2)$; thus $\delta'((r_1, s_2), \alpha)$ takes on the value $(p'', (r_3, s_2))$ because $\iota_{p_2}(\alpha) = (p'', r_3)$.
\end{example}

\begin{proposition}\label{prop:apply-compose-2}
$L(\compose^\ssigma(A)) = \apply_{\ssigma'}(L(A))$, where $\ssigma'$ is given by $\ssigma'(p)=L(\ssigma(p))$.
\end{proposition}

\begin{proof}
By definition, $\apply_{\ssigma'}(L(A))$ is the language consisting of all guarded strings
\[
	\alpha_0 \diamond w_0  \diamond \alpha_1  \diamond \cdots  \diamond \alpha_{n-1}  \diamond w_{n-1}  \diamond \alpha_n
\]
for $\alpha_0 p_0 \alpha_1 \cdots \alpha_{n-1} p_{n-1} \alpha_n \in L(A)$, where each $w_i \in L(\ssigma(p_i))$.
We must show that $L(\compose^\ssigma(A))$ is equal to this language, i.e, that both contain the same guarded strings.
Our strategy will be to prove this by induction on the length of the guarded string in question.
However, for this we need a stronger induction hypothesis, embodied by the following claim.

\medskip\par\emph{Claim 1: } For all $w \in \G(\Sigma_1,T)$
and states $(q_p,q) \in Q_p \times Q$, the following are equivalent:
\begin{enumerate}
    \item
    $w \in L_{\compose^\ssigma(A)}((q_p,q))$
    \item
    $w$ is of the form $v \diamond \alpha_0 \diamond w_0 \diamond \alpha_1 \diamond \cdots \diamond \alpha_{n-1} \diamond w_{n-1} \diamond \alpha_n$
    where $v \in L_{\ssigma(p)}(q_p)$ and there is some guarded string
    $\alpha_0 p_0 \alpha_1 \cdots \alpha_{n-1} p_{n-1} \alpha_n \in L_A(q)$,
    such that each $w_i \in L(\ssigma(p_i))$.
\end{enumerate}
This captures that an accepting run of $\compose^\ssigma(A)$ starting in $(q_p,q)$ consists of an accepting run of $A_p$ starting in state $q_p$, followed by an accepting run of $A$ in state $q$ where each each action $p$ is replaced by a string belonging to $L(\ssigma(p))$.
Claim 1 can now be proved by induction on $w$.
\end{proof}

\Cref{prop:apply-compose} then follows from \Cref{prop:apply-compose-1} and \Cref{prop:apply-compose-2}.

\begin{remark}\label{rem:trivfree-wlog}
Let us say that an instance of composition $\compose^\ssigma(A)$ is \emph{triviality-free} if $\hat{\delta} = \delta$ and $\hat{\iota} = \iota$.
For every automaton $A$ and map $\ssigma$ there is an automaton $A'$ over the same alphabets as $A$ and with at most as many states as $A$, such that $L(\compose^\ssigma(A)) = L(\compose^\ssigma(A'))$ and such that $\compose^\ssigma(A')$ is triviality-free.
Indeed, we can simply replace $A = (Q, \delta, \iota)$ by $A' = (Q, \hat{\delta}, \hat{\iota})$.
For simplicity, we usually work with triviality-free composition when possible in the sequel.
\end{remark}

In preparation of the next subsection, some additional terminology will be helpful.
For two states $r = (q_1, q_2), r' = (q_3, q_4)$ of $\compose^\ssigma(A)$, we say that $r$ and $r'$ are \emph{in the same locale} $q_2 = q_4$, and $q_1$ and $q_3$ belong to the same subautomaton --- that is $q_1$ and $q_3$ are states of $\ssigma(p)$ for some $p \in \Sigma_0$.
A transition is \emph{local} when it is generated by the first case for $\delta'$ in \Cref{def:compose-sigma}.

Intuitively, a local transition in $\compose^\ssigma(A)$ ``stays within the same sub-automaton''.
If $r\transition{\alpha}{p} r'$ is a local transition, then $r$ and $r'$ are in the same locale, but the converse need not hold: a non-local transition could, after exiting a locale, coincidentally re-enter the same locale.

Composed automata enjoy the property that transitions exiting a locale must do so while executing the same actions; in essence, this is true because we work with deterministic automata.
We record this in the following lemma, which will play a technical role soon.

\begin{lemma}\label{lem:non-local-transitions-determinacy}
If $r_1, r_2$ are in the same locale of $\compose^\ssigma(A)$, and $r_1\transition{\alpha}{p} r'_1$ and $r_2\transition{\alpha}{p'} r'_2$ are non-local transitions, then $p=p'$.
\end{lemma}

\begin{proof}
Let $r=(q_1,q)$ and $r'=(q_2,q)$, where $q_1,q_2$ are states of the subautomaton $\ssigma(p)$
for $p\in\Sigma_0$.
Since both transitions are non-local, by the definition of $\delta'$ in \Cref{def:compose-sigma} we have that
\begin{quote}
$\delta_p(q_1,\alpha) = \textsf{accept}$ and $\hat{\delta}(q, \alpha) = (p'', \ldots)$ and $\iota_{p''}(\alpha) = (p, \ldots)$, and, similarly,

$\delta_p(q_2,\alpha) = \textsf{accept}$ and $\hat{\delta}(q, \alpha) = (p'', \ldots)$ and $ \iota_{p''}(\alpha) = (p', \ldots)$
\end{quote}
It then follows that $p=p'$.
\end{proof}

\subsection{Not All Deterministic Languages Arise From Compositions of Small Automata}

We will now establish~\Cref{thm:GKATO-main} by proving the following theorem.
Recall that, in this section, by an \emph{automaton} we will always mean a deterministic KAT automaton.

\begin{theorem}\label{thm:main}
If an automaton $A$ is obtained through (repeated) composition from automata
each having strictly less than $k$ states, then $A$ does not recognize the
language $L_k$.
\end{theorem}

Before proving~\Cref{thm:main}, we will explain why it implies~\Cref{thm:GKATO-main}.

\begin{proof}[Proof of \Cref{thm:GKATO-main}]
Let $\mathcal{O}$ be a finite set of regular control flow operations.
Each operation in $\mathcal{O}$ is defined by a deterministic KAT expression, and thus corresponds to a deterministic KAT automaton, by \Cref{prop:deterministic-iff-language-deterministic,prop:automata-complete}.
Let $k$ be a natural number strictly greater than the maximum the number of states of these (finitely many) deterministic KAT automata.
By \Cref{thm:main}, the deterministic guarded language $L_k$ cannot be expressed by any term built up from test expressions and atomic actions using the operations in $\mathcal{O}$. In other words, $L_k$ is not generated by
$\mathcal{O}$.
\end{proof}

The general intuition behind the proof of \Cref{thm:main} is as follows:
for an automaton to accept $L_k$, it must remember the atom before the most recent action.
This requires $k$ states.
Determinism implies that these $k$ states form a ``cluster'', and these clusters are completely connected (i.e., clique-like), so that they cannot be decomposed.
We will formalize this using a property that can be projected along a composition --- i.e, if a composed automaton satisfies this property, then either the operator or one of the operands does, too.
This will help us pin down the patterns of behavior that cannot appear as a result of repeated composition, but come from the original set of operators.

\begin{definition}\label{def:kdense}
An automaton $(Q, \delta, \iota)$ over alphabets $(\Sigma,T)$ is \emph{$k$-dense} if there exist non-empty subsets $S_1, \ldots, S_k\subseteq Q$, distinct atoms $\alpha_1, \ldots, \alpha_k \in \At_T$ and, for all $1 \leq i, j \leq n$ and $q \in S_i$ an action $p_{qj} \in \Sigma$, such that all of the following hold for all $i,j,m\leq k$:
\begin{itemize}
    \item if $q \in S_i$, then $q$ accepts $\alpha_i$;
    \item if $i \neq j$ and $q \in S_i$, then there exists a $q' \in S_j$ with $q \transition{\alpha_j}{p_{qj}}_A q'$; and
    \item if $q \in S_i$ and $q' \in S_j$ such that $p_{qm} = p_{q'm}$, then $i = j$.
\end{itemize}
\end{definition}

The sets $S_i$ above match the clusters of states hinted at earlier.
The first rule ensures that all states in $S_i$ can accept with atom $\alpha_i$.
Because we limit ourselves to deterministic automata, and the second rule says that all states in $S_i$ transition on $\alpha_j$ with $j \neq i$, all clusters must be disjoint.
The second rule furthermore implies that \emph{all} transitions labeled by $\alpha_j$ starting in $q \in S_i$ must end in $S_j$, and have an action determined by $q$ and $j$.
The third rule says that if two states transition into the same cluster with the same action, then they must also be in the same cluster.

%

We can connect $k$-density to the language $L_k$ as follows.

\begin{lemma}\label{lem:Lk-dense}
Every automaton that recognizes the language $L_k$ must be $k$-dense.
\end{lemma}

\begin{proof}
Let $A=(Q,\delta,\iota)$ be the automaton in question, which, by~\Cref{def:Lk}, has a test alphabet consisting of $t_1, \dots, t_k$, and an action alphabet consisting of $p_1, \dots, p_k$.
We construct a witness to the $k$-density of $A$ as follows.
For each $i\leq k$, let $S_i$ be the set of states $q$ for which there exists a guarded string $w\in L_k$, containing at least one action, ending on atom $\alpha_i$ such
that the accepting run of $w$ ends in state $q$.
Furthermore, for $1 \leq i, j \leq k$ and $q \in S_i$, we choose $p_{qj} = p_i$.

To see that these choices witness $k$-density of $A$, first note that the first and third condition hold by construction.
For the second item, let $q\in S_i$ and let
\(w = x \alpha_{i_n} p_{i_{n-1}} \alpha_{i_n} \in L_k\)
have an accepting run ending in $q$, where $i_n = i$ (and $x$ denotes the relevant prefix).
Now let $j \neq i$ and
\( w' = x \alpha_{i_n} p_{i_{n-1}} \alpha_j p_{i_{n}} \alpha_j \).
Clearly $w'\in L_k$.
Therefore, there must be an accepting run.
Since $A$ is deterministic, this run must extend the accepting run for $w$ by one more transition.
The accepting run for $w'$ must furthermore end in a state $q'$ belonging to $S_j$.
This shows that $q \transition{\alpha_j}{p_{qj}} q'$.
\end{proof}

\begin{lemma}\label{lem:k-states}
If an automaton is $k$-dense, it has at least $k$ states.
\end{lemma}

\begin{proof}
Let $S_1, \ldots, S_k$ be witnessing subsets of the state space.
By definition, $S_1, \ldots, S_k$ are non-empty.
They are also disjoint: suppose $r\in S_i\cap S_j$ for some $i\neq j$.
Since $r\in S_i$, we have that $\delta(r, \alpha_i) = \mathsf{accept}$.
Since $r\in S_j$, we have that $r \transition{\alpha_i}{p_j} s$ for some $s\in S_i$. This contradicts
determinacy.
\end{proof}

\begin{lemma}\label{lem:dense-tau}
If $\compose_\ttau(A)$ is  $k$-dense, then so is $A$.
\end{lemma}

\begin{proof}
By definition, there are
non-empty subsets $S_1, \ldots, S_k\subseteq Q$, distinct atoms $\alpha_1, \ldots, \alpha_k \in \At_{T_1}$ and, for all $1 \leq i, j \leq k$ and $q \in S_i$ an action $p_{qj} \in \Sigma$ satisfying the conditions of Definition~\ref{def:kdense}.
Let $\alpha'_i =\ttau^{-1}(\alpha_i)$.
It now suffices to use $\alpha'_i$ instead of $\alpha_i$ to see that $A$ is $k$-dense.
\end{proof}

It remains to prove a similar lemma for $\compose^\ssigma(A)$.
We first  prove an ``all-or-nothing'' lemma.

\begin{lemma}\label{lem:all-or-nothing}
Suppose that $\compose^\ssigma(A)$ is $k$-dense, and let
$S_1, \ldots, S_k$ be witnessing subsets of the state space of $\compose^\ssigma(A)$.
One of the following two conditions holds:
\begin{enumerate}
\item there do not exist elements $r\in S_i, r'\in S_j$ with $i\neq j$ that are in the same locale, or
\item there are $r_1\in S_1, \ldots, r_k\in S_k$ that are all in the same locale.
\end{enumerate}
\end{lemma}

\begin{proof}
Suppose, for the sake of a contradiction, that neither condition holds.
Then there are
(1)~$i, j \leq k$ with $i \neq j$, and $r \in S_i$, $r'\in S_j$ such that $r$ and $r'$ are in the same locale, and
(2)~$\ell\leq k$ such that there is no $r''\in S_\ell$ in the same locale as $r$ and $r'$.
The latter means that $i, j \neq \ell$.
Because $\compose^\ssigma(A)$ is $k$-dense, we can obtain $s,s'\in S_\ell$ such that  $r\transition{\alpha_\ell}{p_{i\ell}}s$ and $r'\transition{\alpha_\ell}{p_{j\ell}} s'$.
Because $r$ (resp.\ $r'$) is in a different locale than $s$ (resp.\ $s'$), these are non-local transitions; thus, by Lemma~\ref{lem:non-local-transitions-determinacy}, $p_{i\ell}=p_{j\ell}$.
By the third condition of $k$-density, we then find that $i = j$, contradicting that $i \neq j$.
\end{proof}

\begin{lemma}\label{lem:dense-sigma}
If $\compose^\ssigma(A)$ is triviality-free and  $k$-dense, then either $A$ is $k$-dense or $\ssigma(p)$ is $k$-dense for some $p$.
\end{lemma}

\begin{proof}
Let $A = (Q, \delta, \iota)$ be a deterministic KAT automaton over $(\Sigma_0,T)$, and for $p \in \Sigma_0$, let $\ssigma(p) = (Q_p,\delta_p,\iota_p)$ be a deterministic KAT automaton over $(\Sigma_1, T)$.
We write $A' = (Q', \delta', \iota')$ for $\compose^\ssigma(A)$.
Because $A'$ is $k$-dense, we have $S_1, \dots, S_k \subseteq Q'$ and for all $i, j \leq k$ and $q \in S_i$ we have $p_{qj} \in \Sigma_1$, satisfying the conditions of $k$-density.
By \Cref{lem:all-or-nothing}, we can distinguish two cases.

Case (1): no $q \in S_i $ and $q' \in S_j$ for $i \neq j$ are in the same locale.
Now all transitions between states in $S_i$ and $S_j$ are non-local.
We claim that $A$ is $k$-dense.
For $i\leq k$, let $S'_i = \{ s \mid (r, s) \in S_i \}$, and for $1 \leq i, j \leq k$ and $s \in S_i'$, choose $p_{sj}'$ such that $\delta(s, \alpha_j) = (p_{sj}', -)$.
These witness the $k$-density of $A$:
\begin{enumerate}
    \item
    If $s \in S'_i$, then $(r, s) \in S_i$ for some $r$.
    Because $A'$ is $k$-dense, we then know that $(r, s)$ accepts $\alpha_i$ in $\compose^\ssigma(A)$; by definition of composition, $s$ must also accept $\alpha_i$ in $A$.
    \item
    If $s \in S'_i$, then $(r, s) \in S_i$ for some $r$, as before.
    Because $A'$ is $k$-dense, we then know that $(r, s) \transition{\alpha_j}{p_{(r,s)j}}_{A'} (r', s')$ for some $(r', s') \in S_j$.
    Because this is a non-local transition, by definition of (triviality-free) composition it must be the case that $s \transition{\alpha_j}{p_{sj}'}_A s'$.
    \item
    Finally, let $s \in S'_i$, $s' \in S'_j$ and $m \leq k$ with $i, j \neq m$ and $p_{sm}' = p_{s'm}'$.
    There exist $r$ and $r'$ such that $(r, s) \in S_i$ and $(r', s') \in S_j$.
    By $k$-density of $A'$, there also exist $q, q' \in S_m$ such that $(r, s) \transition{\alpha_m}{p_{(r, s)m}}_{A'} q$ and $(r', s') \transition{\alpha_m}{p_{(r', s')m}}_{A'} q'$.
    But because these are non-local, and the actions on the $\alpha_m$-transitions exiting $s$ and $s'$ in $A$ are $p_{sm}' = p_{s'm}'$, by definition of composition $p_{(r, s)m}$ and $p_{(r',s')m}$ are determined by the initialization of $\ssigma(p_{sm}') = \ssigma(p_{s'm}')$.
    Thus, $p_{(r,s)m} = p_{(r',s')m}$, and by $k$-density of $A'$, it follows that $i = j$ as desired.
\end{enumerate}

Case (2): there are $q_1 \in S_1, \ldots, q_k \in S_k$ that are all in the same locale.
This means that there exists a $p \in \Sigma_0$ and an $s \in Q$, as well as $r_1, \dots, r_k \in Q_p$ such that $q_i = (r_i, s)$ for all $i \leq k$.
Our claim is that $\ssigma(p)$ is $k$-dense, although building a witness for this will require some care.

We start by making two observations.
First, note that for $i \leq k$, we have that the $s$ above accepts $\alpha_i$ in $A$.
This follows by definition of composition, and the fact $(r_i, s) \in S_i$, so the latter accepts $\alpha_i$ in $A'$.
Second, if $i, j \leq k$ with $i \neq j$, and $(r, s) \in S_i$ as well as $q \in S_j$ such that $(r, s) \transition{\alpha_j}{p_{(r, s)j}}_{A'} q$, then the latter is a local transition --- in particular, $q$ is of the form $(r', s)$ with $r' \in Q_p$.
This follows by the fact that $s$ accepts $\alpha_j$ (our first observation), and the definition of composition.

With these in mind, we can now choose our witness for $k$-density of $\ssigma(p)$ as follows.
For $i \leq k$, we set $S_i' = \{ r \in Q_p \mid (r, s) \in S_i \}$, where $s$ is the fixed state obtained above.
Furthermore, for $i, j \leq k$ with $i \neq j$ and $r \in S_i'$, we choose $p_{rj}' = p_{(r, s)j}$.
We can now check the conditions.
\begin{enumerate}
    \item
    If $r \in S_i'$, then $(r, s) \in S_i$, and hence $r$ accepts $\alpha_i$ in $\ssigma(p)$, because $(r, s)$ accepts $\alpha_i$ in $A'$.
    \item
    If $r \in S_i'$ and $i \neq j$, then by the second observation we know that there exists an $r' \in Q_p$ such that $(r', s) \in S_j$, and $(r, s) \transition{\alpha_j}{p_{(r, s)j}}_{A'} (r', s)$ is a local transition.
    We then find that $r' \in S_j'$ by construction, and $r \transition{\alpha_j}{p_{rj}'}_{\ssigma(p)} r'$, because the former transition is local and $p_{rj}' = p_{(r, s)j}$.
    \item
    Finally, let $r \in S_i'$ and $r' \in S_j'$, with $m \leq k$ such that $p_{rm}' = p_{r'm}'$.
    Then $p_{(r, s)m} = p_{(r', s)m}$, and because $A$ is $k$-dense with $(r, s) \in S_i$ and $(r', s) \in S_j$, it follows that $i = j$ as desired.
    \qedhere
\end{enumerate}
\end{proof}

Since we defined $\compose^\ssigma_\ttau(A)$ as $\compose^\ssigma(\compose_\ttau(A))$,
the above two lemmas imply:

\begin{lemma}\label{lem:dense}
If a triviality-free composition $\compose^\ssigma_\ttau(A)$ is  $k$-dense, then either $A$ is $k$-dense, or $\ssigma(p)$ is $k$-dense for some $p$.
\end{lemma}

\Cref{thm:main} follows immediately from \Cref{lem:Lk-dense,lem:k-states,lem:dense}.
Note that, when an automaton is obtainable through (repeated) composition from automata having strictly less than $k$ states, then by \Cref{rem:trivfree-wlog} it is also obtainable in the same way using only triviality-free compositions.

\begin{remark}
Call a GKAT expression \emph{$k$-sparse} if its language is recognized by a
deterministic KAT automaton that is not $k$-dense. It is not hard to see that
$\textsf{assert $b$}$, $p;p'$, $\Ifthenelse{b}{p}{p'}$ and $\while{b}{p}$ (the basic operations
of GKAT) are 2-sparse. Hence, every GKAT expression is 2-sparse.
\end{remark}

\section{Further Results}%
\label{sec:further}

We now present two minor results that help put our main theorem in context.
First, we construct an infinite family of regular control flow operations that can express all regular control flow operations.
Next, we  show that the deterministic fragment of Kleene algebra \emph{without} tests \emph{is} finitely generated, and we offer some thoughts about other variants of Kleene algebra.

\subsection{A Positive Result}%
\label{sec:positive}

The set of all regular control flow operations is trivially expressively complete, i.e., it generates the entire deterministic fragment of KAT\@.
In this section, we improve on this slightly by exhibiting a more restricted infinite set of regular control flow operations that enjoys the same property.

\begin{definition}[Free Automaton]
The \emph{free automaton with $k$ states} is the deterministic KAT automaton $A_k=(Q,\delta,\iota)$ over actions $\Sigma_k=\{ p_1, \ldots, p_k \}$ and tests $T_k=\{t_{i,j}\mid 0 \leq i,j \leq k \}$, where
\begin{mathpar}
    Q=\{q_i\mid 1 \leq i \leq k \}
    \and
    \delta(q_i,\alpha) =
        \begin{cases}
        \textsf{accept}  & \text{if $\alpha \leq b_{i,0}$} \\
        (p_j,q_j) & \text{if $\alpha \leq b_{i,j}$} \\
        \textsf{reject}  & \text{otherwise}
        \end{cases}
    \and
    \iota(\alpha) =
        \begin{cases}
        \textsf{accept}  & \text{if $\alpha \leq b_{0,0}$} \\
        (p_i,q_i) & \text{if $\alpha \leq b_{0,j}$} \\
        \textsf{reject} & \text{otherwise}
        \end{cases}
\end{mathpar}
where, for $0 \leq i, j \leq k$, we define $b_{i,j} \in \BA(T_k)$ to be the test that holds when $t_{i,j}$ is true, and for all $m < i$ we have that $t_{m,j}$ is \emph{not} true.
Note that this makes the $b_{i,j}$ mutually exclusive for fixed $i$.
\end{definition}

One way to think of the control flow operation $O_k$ is as providing the means to define a collection of $k$ auxiliary functions $f_1, \ldots, f_k$ by mutual tail recursion, as in the following pseudocode:
\[
\parbox{50mm}{\tt\raggedright
function $f_1()$ $\{$ \\
\mbox{~~} if $b_{1,0}$ then $\{e_{1,0}\}$ \\
\mbox{~~} else if $b_{1,1}$ then $\{e_{1,1};f_1()\}$ \\
\mbox{~~} \ldots \\
\mbox{~~} else if $b_{1,k}$ then $\{e_{1, k}; f_k()\}$ \\ 
\mbox{~~} else assert false \\
$\}$
}
\hspace{10mm}\cdots\hspace{10mm}
\parbox{50mm}{\tt\raggedright
function $f_k()$ $\{$ \\ 
\mbox{~~} if $b_{k,0}$ then $\{e_{k,0}\}$ \\
\mbox{~~} else if $b_{k,1}$ then $\{e_{k,1};f_1()\}$ \\
\mbox{~~} \ldots \\
\mbox{~~} else if $b_{k,k}$ then $\{e_{k,k}; f_k()\}$ \\ 
\mbox{~~} else assert false \\
$\}$
}
\]
This makes it clear that the program $F$ from~\Cref{definition:repeat-while-changes-semantics},
can be expressed using $O_2$ as an operator.

A free automaton on $k$ states can thus be seen as an ``automaton template'', with transitions between any pair of states and the option to accept anywhere.
Recovering a given automaton from a free automaton (up to equivalence) using composition is a matter of choosing the right parameters.

\begin{proposition}%
\label{prop:free-automata}
For every deterministic KAT automaton $A$ over $(\Sigma,T)$, there exists a $k>0$ and maps $\ssigma, \ttau$, such that $L(A)=L(\compose^\ssigma_\ttau (A_k))$, and for all $p\in\Sigma_k$, $L(\ssigma(p)) = L(p')$ for some $p'\in\Sigma$.
\end{proposition}

\begin{proof}
Let $A=(Q,\delta,\iota)$, let $Q=\{q_1, \ldots, q_k\}$ and let $\Sigma=\{p_1, \ldots, p_n\}$.
We may assume without loss of generality that, for each $q_i\in Q$, there is a unique $p_j\in\Sigma$ such that all transitions (including initial transitions) leading to $q_i$ involve $p_j$.
This can easily be ensured by duplicating each state $n$ times, replacing each transition $\delta(q_i,\alpha)=(p_j,q_{i'})$ by all transitions of the form $\delta(\widehat{q_i},\alpha)=(p_j,\widehat{q_{i'}})$ where $\widehat{q_i}$ is any copy of $q_i$ and $\widehat{q_{i'}}$ is the $j$-th copy of $q_{i'}$, and, similarly for initial transitions.
In this way, for every state, all transitions leading to it involve the same action.

We will now construct $\ssigma, \ttau$ such that $L(A)=L(\compose^\ssigma_\ttau (A_k))$.
Recall that $A_k$ is an automaton over action and test alphabet $\Sigma_k, T_k$. We define $\ssigma$ and $\ttau$
as follows, where $1\leq i,j\leq n$:
\begin{itemize}
\item $\ssigma(p_j)=A_p$, where $A_p$ is the single-state deterministic KAT automaton with language $L(p)$, and $p\in\Sigma$ is the unique action associated with the state $q_j$ of $A$.
\item $\ttau(t_{0,0})$ is the disjunction of all atoms $\alpha\in\At_T$ for which it holds that
   $\iota(\alpha)=\textsf{accept}$
\item $\ttau(t_{0,i})$ is the disjunction of all atoms $\alpha\in\At_T$ for which it holds that
   $\iota(\alpha)=(\ldots,q_i)$
\item $\ttau(t_{i,0})$ is the disjunction of all atoms $\alpha\in\At_T$ for which it holds that
   $\delta(q_i,\alpha)=\textsf{accept}$
\item $\ttau(t_{i,j})$ is the disjunction of all atoms $\alpha\in\At_T$ for which it holds that
   $\delta(q_i,\alpha)=(\ldots, q_j)$
\end{itemize}
Now $\compose^\ssigma_\ttau (A_k)$ is isomorphic to $A$, and thus it defines the same language.
\end{proof}

By \Cref{prop:automata-complete}, each $A_k$, being a deterministic KAT automaton, corresponds to a deterministic KAT expression $e_k$, and hence to a regular control flow operation $O_k$.
It follows from~\Cref{prop:free-automata} that the infinite set $\mathcal{O}=\{O_k\mid k>0\}$ generates the entire deterministic fragment of KAT\@.

\begin{remark}
According to our result in~\Cref{sec:main}, every set of regular control flow operations that generates the entire deterministic fragment of KAT must include operations whose corresponding automaton is $k$-dense for arbitrarily large $k>0$. The free automaton with $k$ states is in general not $k$-dense, but can easily be seen to be $\lfloor\sqrt{k}\rfloor$-dense. Therefore, the above set $\mathcal{O}=\{O_k\mid k>0\}$ does indeed include operations whose corresponding automaton is $k$-dense for arbitrarily large $k>0$.
\end{remark}

\subsection{Other Variants of Kleene Algebra}%
\label{sec:variations}

Consider the subset of regular control flow operations that do not have any tests as parameters.
These KAT expressions can equivalently be regarded as expressions in the language of plain Kleene algebra \emph{without tests} (KA), and so we refer to them as the \emph{deterministic fragment of KA}.
It turns out that this set of operations \emph{is} finitely generated, which we record below.

\begin{theorem}
The deterministic fragment of KA is generated by the sequential composition operator together
with the constant operations $\mathsf{true}$ and $\mathsf{false}$.
\end{theorem}

\begin{proof}
Clearly, every term built from atomic programs, $\mathsf{true}$ and $\mathsf{false}$, using
sequential composition, is deterministic. It therefore only remains to prove the converse, i.e.,
that every deterministic KA-term can be equivalently written in this form.

Let $e$ be any deterministic KA-term. Now, $e$ can be viewed as a regular expression over $\Sigma$.\footnote{Note that every variable-free test expression $b$ is equivalent to $\mathsf{true}$ or to $\mathsf{false}$ and hence can be treated as 0 or as 1 respectively,  where 0 is the regular expression that defines the empty language and 1 is the regular expression that defines the language consisting of the empty word.}
If the regular language of $e$, denoted $L_R(e)$, contains two distinct strings $s, s'$, then it is easy to see that $e$ is non-deterministic (it suffices to turn $s$ and $s'$ into guarded strings using an arbitrary atom).
Thus, either $L_R(e) = \emptyset$ or $L_R(e) = \{ s \}$ for some string $s$.
But then either $L(e)=\emptyset$, or the guarded strings in $L(e)$ are precisely all the guarded strings of the form $\alpha_1 p_1 \alpha_2 \cdots \alpha_n p_n \alpha_{n+1}$, where $s = p_1p_2 \cdots p_n$.
In the former case, $e$ is equivalent to the KA-term $\mathsf{false}$.
In the latter case, $e$ is equivalent to a KA-term $p_1 \cdot \ldots \cdot p_n$ if $n > 0$, and $e$ is equivalent to $\mathsf{true}$ otherwise. 
\end{proof}

In other words, the presence of tests in the language plays a crucial role in our main
result for KAT\@.
We conjecture that our results regarding the deterministic fragment of KAT can be adapted to the deterministic fragment of \emph{Kleene Algebra with Domain (KAD)}~\cite{KAD} and of \emph{Propositional Dynamic Logic} (or, equivalently, Kleene Algebra with Antidomain~\cite{sedlar-2023}).

\section{Discussion}%
\label{sec:discussion}
Our main result can be summarized informally, as follows: \emph{compositional control flow languages consisting of a finite set of operators, such as GKAT, are inherently expressively incomplete.}
Nevertheless, we believe that this result generalizes and adds more depth to earlier insights about control flow.
It does not invalidate GKAT as a topic worth studying.
On the contrary, we think there are still interesting matters to be settled.
If anything, we take this result as encouragement to study the expressiveness of different kinds of control flow primitives, such as exceptions or other limited forms of non-local control flow.
We now briefly go over a number of relevant works in this area.

\emph{Program schematology}, i.e., the study of control flow, traces its origins to the aforementioned result by \citet{BohmJacopini1966} about realizing all control flow using composition and iteration.
The transformation from op.\ cit.\ uses a variable to keep track of the current decision point within the program --- not unlike one of the ``workarounds'' we considered to implement the program $F$.
The algorithm proposed by \citet{erosa-hendren-1994}, which gradually eliminates non-local control flow from a program, is perhaps more practical; see also~\cite{casse-etal-2002} for an implementation.


\citet{knuth-floyd-1971} first connected control flow to regular languages, and relay a proof by Hopcroft that can be viewed as a first inexpressivity result.
\citet{ashcroft.manna:translation} studied automatic goto-elimination, and showed that variables are necessary to achieve all control flow; in fact, the program $F$ discussed earlier was partly inspired by them.
\citet{peterson-kasami-tokura-1973} considered the expressivity of multi-level exits from loops, showing that single-level loop exits (like \emph{break}) do not suffice to express all control flow.
Taking a more systematic approach, \citet{Kosaraju} showed that there exists a hierarchy of expressivity of multi-level break statements: the ability to break $n+1$-deeply nested loops affords more expressivity than being able to break $n$-deeply nested loops.

Our result, like Kosaraju's, implies a hierarchy to deterministic control flow constructs.
However, we focused on program \emph{composition}; because non-local control flow (e.g., \emph{break}) makes sense only in the context of a larger program, it is outside of the scope of our techniques.
We are excited about KAT as a \emph{lingua franca} for rigorously studying expressivity of control flow constructs --- though we are hardly the first to do so~\cite{angus-kozen-2001,KozenTseng,grathwohl-etal-2014,kozen-2008}.
The interested reader is directed to~\cite{zhang-etal-2024} for further research on how to realize non-local control flow within deterministic KAT, and efficiently decide equivalence\@.

\citet{milner-1984} investigated regular expressions \emph{up to bisimilarity}.
A full characterization of the expressivity of this semantics was left open until recently~\cite{grabmayer-2022,grabmayer-fokkink-2020}.
It turns out that the expressivity of a certain fragment of GKAT is intricately related to this work~\cite{kappe-schmid-silva-2023}.
Underlying the efforts in this area is a direct connection of expressivity to completeness of equational reasoning systems for programs; put briefly, if the automata expressed by a certain fragment can be precisely characterized, then a completeness theorem is usually not too far away.
In future work, we hope to consider the expressivity problem of the full calculus of GKAT, and by extension its open completeness problem~\cite{GKAT,GKATBisim,kappe-schmid-silva-2023}, using the techniques developed in this article.

\citet{nivat-1972,nivat-1974} and subsequent authors studied \emph{recursive program schemes (RPSs)} as a formalism for reasoning about the control flow of (mutually) recursive programs.
Viewed through this lens, deterministic KAT automata are \emph{finite}, \emph{first-order} and \emph{linear tail-recursive} (uninterpreted) RPSs. See also \citet{Garland1973} for relationships between RPSs and automata.
Our main theorem then says that there is no finite set of such RPSs into which all other such RPSs can be hierarchically decomposed.
It would  be interesting to study inductive characterizations of the operators induced by deterministic KAT automata, as was done for RPSs by \citet{debakker-meertens-1972}.

Finally, from the perspective of the relational semantics, \emph{intersection} and  \emph{subtraction} are natural functionality-preserving operations on binary relations.
They are, however, not regular (cf.\ \Cref{lem:intersection-nonregular}).
Consequently, our results do cover extensions of GKAT with these operations.
We wonder if some version still holds for KAT extended with these operators.

\begin{acks}
B.~ten~Cate was supported by \grantsponsor{baldermc}{the European Union’s Horizon 2020 research and innovation programme}{https://www.eeas.europa.eu/eeas/horizon-2020_en} under grant no.\ \grantnum{baldermc}{101031081} (LLAMA).
T.~Kappé was partially supported by \grantsponsor{tobiasmc}{the European Union’s Horizon 2020 research and innovation programme}{https://www.eeas.europa.eu/eeas/horizon-2020_en} under grant no.\ \grantnum{tobiasmc}{101027412} (VERLAN), and partially by \grantsponsor{tobiasveni}{the Dutch research council (NWO)}{https://www.nwo.nl/} under grant no.\ \grantnum{tobiasveni}{VI.Veni.232.286} (ChEOpS).

We thank Larry Moss for a helpful discussion during the revision stage.
\end{acks}

\printbibliography%

\end{document}
\endinput